\documentclass{article}
\usepackage{graphicx}
\usepackage[T1]{fontenc}
\usepackage[utf8]{inputenc}
\usepackage[english]{babel}
\usepackage[margin=0.8in]{geometry}
\usepackage{microtype}
\usepackage{amsmath, amssymb}
\usepackage{amsthm}
\usepackage{mathtools}
\usepackage{booktabs}
\usepackage{tabularx}
\usepackage[hyphens]{url}
\usepackage{float}
\usepackage{hyperref}
\usepackage{enumitem}
\usepackage{bm}
\usepackage{xcolor}
\theoremstyle{definition}
\newtheorem{theorem}{Theorem}[section]
\newtheorem{lemma}[theorem]{Lemma}
\newtheorem{proposition}[theorem]{Proposition}
\newtheorem{corollary}[theorem]{Corollary}
\theoremstyle{definition}
\newtheorem{definition}[theorem]{Definition}
\newtheorem{remark}[theorem]{Remark}
\newtheorem{example}[theorem]{Example}
\newtheorem{notation}[theorem]{Notation}

\DeclareMathOperator{\id}{id}
\newcommand{\Lx}{\mathbb{L}[x_1,\dots,x_m;\sigma_1,\dots,\sigma_m]}
\DeclareMathOperator{\ev}{ev}
\DeclareMathOperator{\Gal}{Gal}
\DeclareMathOperator{\Span}{Span}
\DeclareMathOperator{\Mat}{Mat}
\DeclareMathOperator{\SRRM}{SRRM}
\DeclareMathOperator{\LRM}{LRM}
\DeclareMathOperator{\LT}{LT}
\DeclareMathOperator{\LM}{LM}
\DeclareMathOperator{\rk}{rk}
\DeclareMathOperator{\srk}{srk}
\DeclareMathOperator{\N}{N}
\newcommand{\lex}{<_\mathrm{lex}}
\newcommand{\End}{\mathrm{End}}
\newcommand{\Hom}{\mathrm{Hom}}

\newcommand{\LK}{\mathbb L/\mathbb K}

\renewcommand{\u}{\bm{u}}
\renewcommand{\v}{\bm{v}}
\newcommand{\e}{\bm{e}}

\DeclareMathOperator{\GL}{GL}
\DeclareMathOperator{\ord}{ord}
\DeclareMathOperator{\Tr}{Tr}
\DeclareMathOperator{\diag}{diag}

\newcommand{\F}{\mathbb F}

\title{Sum-rank analogues of Reed--Muller codes \\
over Abelian Galois extensions}
\author{Alessandro Neri\thanks{Dipartimento di Matematica e Applicazioni ``R.~Caccioppoli'', Universit\`a degli Studi di Napoli Federico II, Italy. \texttt{alessandro.neri@unina.it}}
\and Giuseppe Zampella\thanks{Dipartimento di Matematica e Applicazioni ``R.~Caccioppoli'', Universit\`a degli Studi di Napoli Federico II, Italy. \texttt{giuseppe.zampella@unina.it}}}

\begin{document}
\maketitle

\begin{abstract}
Let $\mathbb L/\mathbb K$ be a finite Abelian Galois extension with a fixed decomposition $\Gal(\mathbb L/\mathbb K)=\langle\sigma_1\rangle\times\cdots\times\langle\sigma_m\rangle$. We show that the quotient of the multivariate skew polynomial ring $\mathbb L[x_1,\ldots,x_m;\sigma_1,\ldots,\sigma_m]$ by suitable central polynomials $H_1(x_1),\ldots,H_m(x_m)$ is isomorphic to a product of $t$ copies of $\End_{\mathbb K}(\mathbb L)$, so that one skew polynomial represents all blocks of a sum-rank metric space. The sum-rank weight of an element equals the rank of an associated Dickson matrix and of a generalized Moore matrix. Using this, we define sum-rank Reed--Muller-type codes and show that their dimension and minimum distance are given by the formulas for affine Cartesian codes with side lengths $\deg H_1,\ldots,\deg H_m$; the same holds for codes supported on lower sets. The family contains as special cases linearized Reed--Solomon codes, Gabidulin codes, rank-metric Reed--Muller codes, and affine Cartesian codes. Over finite fields, every linearized Reed--Muller code of Berardini and Caruso is, after a monomial change of variables, a monomial code in our framework. This gives the dimension of all these codes, and their exact minimum distance in every degree whenever the support is a lower set.
\end{abstract}



\section{Introduction}

\paragraph{Background and motivation.}
Polynomial evaluation is a standard method for constructing error-correcting codes. Reed--Solomon codes arise from evaluations of univariate polynomials, and their minimum distance follows from the bound on the number of roots of a nonzero polynomial. Generalized Reed--Muller and affine Cartesian codes use multivariate polynomials evaluated on finite grids \cite{KLP68,GT2013,LRMV2014}. Their dimensions are obtained by counting admissible monomials, while their minimum distances are related to bounds on the number of zeros of a polynomial, including the Alon--F\"uredi theorem and the Schwartz--Zippel lemma \cite{AF,S,Z,BCPS2018,CFS2017}.

These families use the Hamming metric, which counts the coordinates at which two words differ. The rank and sum-rank metrics arise in settings where errors have a matrix structure, notably in network coding. In the \emph{rank metric} \cite{Delsarte1978,G85,Roth1991}, words are matrices over a field $\mathbb K$, and the distance between two matrices is the rank of their difference. The \emph{sum-rank metric} includes both the rank and the Hamming metrics. Its ambient space consists of $t$-tuples $A=(A_1,\ldots,A_t)$ of matrices over $\mathbb K$ of prescribed sizes, called blocks. The \emph{sum-rank weight} is
\[
  \srk(A):=\sum_{j=1}^t\rk_{\mathbb K}(A_j).
\]
The distance between $A$ and $B$ is $\srk(A-B)$; for a linear code, the minimum sum-rank distance is therefore the minimum weight of a nonzero codeword. Taking $t=1$ gives the rank metric, while taking each block to be a single column gives the Hamming metric. The sum-rank metric appeared in space--time coding \cite{LK2005} and in multishot network coding \cite{NU2010}. It has since been used for reliable and secure multishot network coding \cite{MPK2019multishot} and for locally repairable codes with maximal recoverability \cite{MPK2019LRC}. The size and the minimum distance of a sum-rank metric code satisfy a Singleton-type bound, and codes attaining it are called maximum sum-rank distance (MSRD) codes \cite{M-P2018,BGR2021}. Further bounds on the parameters of sum-rank metric codes are given in \cite{BGR2021,AKR2024}. We refer to \cite{M-P2022} for additional background.

Polynomial constructions also extend to the rank and sum-rank metrics. Gabidulin codes use Frobenius powers in place of ordinary powers and play the role of Reed--Solomon codes in the rank metric \cite{G85}; they can be defined over any cyclic Galois extension, also in characteristic zero \cite{ALR2018}. Univariate skew polynomials give linearized Reed--Solomon codes, which attain the sum-rank Singleton bound and include Reed--Solomon and Gabidulin codes as special cases \cite{M-P2018}.  Other families of sum-rank metric codes include twisted linearized Reed-Solomon codes \cite{N}, Hamming and simplex codes \cite{MP2020Hamming} and product codes \cite{ALNW2022}.

Skew polynomial rings also provide algebraic models of the ambient spaces. Augot, Couvreur, Lavauzelle, and Neri constructed multivariate rank-metric Reed--Muller codes over finite Abelian Galois extensions \cite{ACLN}. They describe the rank-metric ambient space as a quotient of a multivariate skew polynomial ring, using the skew group algebra, and obtain distance bounds analogous to those of Alon--F\"uredi and Schwartz--Zippel. Applying this construction separately to $t$ sum-rank blocks gives $t$ polynomial classes rather than a single class representing all the blocks (see the end of Section~\ref{sec:preliminaries}). For cyclic Galois extensions, quotients of univariate skew polynomial rings are used in \cite{N,NS2025} to represent matrix algebras and their direct sums, so that all blocks are represented at once. For other approaches based on multivariate skew polynomials, see \cite{GU2019, MP2022b}. 

Berardini and Caruso introduced linearized Reed--Muller codes in the sum-rank metric over finite fields by evaluating multivariate Ore polynomials at families of endomorphisms \cite{BC2025}. They also showed that many of these codes embed into their linearized algebraic-geometry codes \cite{BC2024}; explicit and asymptotically good constructions of the latter are given in \cite{BBGM2026}. We compare their construction with ours in Section~\ref{subsec:comparison-LRM}.

\paragraph{Our contributions.}
We extend the quotient-algebra construction of \cite{ACLN} to represent all blocks of a sum-rank ambient space by a single polynomial class, and use this representation to define sum-rank Reed--Muller-type codes. Let
$\mathbb L/\mathbb K$ be a finite Abelian Galois extension of degree $N$,
and fix an internal direct-product decomposition
\[
G=\Gal(\mathbb L/\mathbb K)
=\langle\sigma_1\rangle\times\cdots\times\langle\sigma_m\rangle,
\qquad \ord(\sigma_i)=n_i,
\qquad N=\prod_{i=1}^m n_i.
\]
For each $i$, choose $t_i$ nonzero elements
$\alpha_{i,1},\ldots,\alpha_{i,t_i}$ in the intermediate field fixed by
all $\sigma_h$ with $h\neq i$, and assume that their relative norms
$\lambda_{i,j}\in\mathbb K^*$ are pairwise distinct as $j$ varies. Set
\[
t:=\prod_{i=1}^m t_i,
\qquad
\nu_i:=n_it_i,
\qquad
\bm\nu:=(\nu_1,\ldots,\nu_m),
\qquad
H_i(x_i):=\prod_{j=1}^{t_i}(x_i^{n_i}-\lambda_{i,j}).
\]
We also write $\Lambda_i:=\{\lambda_{i,1},\dots,\lambda_{i,t_i}\}$ for $i=1,\dots,m$. The first main result, the Representation Theorem (Theorem~\ref{representation theorem}), identifies a quotient of the multivariate skew polynomial ring with the full sum-rank ambient algebra:
\[
\frac{\mathbb L[x_1,\ldots,x_m;\sigma_1,\ldots,\sigma_m]}
{(H_1(x_1),\ldots,H_m(x_m))}
\cong
\big(\mathbb L[G]\big)^t
\cong
\big(\mathbb K^{N\times N}\big)^t,
\]
where the last identification depends on the choice of a $\mathbb K$-basis of $\mathbb L$.
The class of a skew polynomial $F$ corresponds to the $t$-tuple obtained by evaluating $F$ at the commuting operators
$\alpha_{1,j_1}\sigma_1,\ldots,\alpha_{m,j_m}\sigma_m$, one for each choice of indices $1\leq j_i\leq t_i$.
Each evaluation acts on $\mathbb L$ as a $\mathbb K$-linear endomorphism. Hence, one polynomial class determines all $t$ matrix blocks, and its sum-rank weight is the sum of their ranks, independently of the chosen basis.

The quotient has a standard monomial basis indexed by
$\Delta(\bm\nu)=\prod_{i=1}^m\{0,\ldots,\nu_i-1\}$.
This leads to the sum-rank Reed--Muller-type code of order $r$,
\[
\SRRM_{\mathbb L/\mathbb K}(r,m)
=
\Span_{\mathbb L}
\left\{
\bm x^{\bm u}\bmod(H_1,\ldots,H_m)
\ \middle|\
\bm u\in\Delta(\bm\nu),\ |\bm u|\leq r
\right\},
\]
where $|\bm u|=u_1+\cdots+u_m$.

The sum-rank length is
\[
Nt=\prod_{i=1}^m\nu_i,
\]
and the dimension is the number of admissible monomials:
\[
k(r,\bm\nu)
=
\#\left\{\bm u\in\Delta(\bm\nu)\mid |\bm u|\leq r\right\}
=
\sum_{h=0}^r[z^h]\prod_{i=1}^m\frac{1-z^{\nu_i}}{1-z}.
\]
Here $[z^h]$ denotes coefficient extraction.
For $0\leq r<\sum_{i=1}^m(\nu_i-1)$, after reordering the variables so
that $\nu_1\geq\cdots\geq\nu_m$, write
\[
r=\sum_{i=s+1}^m(\nu_i-1)+\ell,
\qquad 0\leq\ell<\nu_s-1.
\]
Our exact minimum-distance formula (Theorem~\ref{thm:minimum-distance}) is
\[
d_{\min}
\bigl(\SRRM_{\mathbb L/\mathbb K}(r,m)\bigr)
=
(\nu_s-\ell)\prod_{i=1}^{s-1}\nu_i.
\]
Thus, the dimension and minimum distance have the same combinatorial formulas as those of affine Cartesian codes, with side lengths $\nu_i=n_it_i$.

To establish the distance bounds, we use the Dickson matrix of right multiplication by a polynomial class and a generalized Moore matrix adapted to the Abelian Galois decomposition and the chosen elements $\alpha_{i,j}$. We prove that these matrices have the same rank over $\mathbb L$, equal to the sum-rank weight of the corresponding codeword. The weight can therefore be computed from the image of a multiplication map, or equivalently from a colon ideal. An argument with leading monomials gives an Alon--F\"uredi-type lower bound and a Schwartz--Zippel-type consequence. The lower bound is attained by products of univariate right divisors of the polynomials $H_i$, one in each variable.

The distance formula extends to monomial codes supported on lower sets. For every nonempty lower set $D\subseteq\Delta(\bm\nu)$, that is, a set closed under componentwise decrease, the associated code has dimension $\#D$ and exact minimum distance (Theorem~\ref{thm:lower-set-distance})
\[
\min_{\bm u\in D}\prod_{i=1}^m(\nu_i-u_i).
\]
This also covers weighted-degree constructions, which arise in the comparison with lattice-simplex codes.

The construction recovers several familiar code families, summarized in Table~\ref{tab:specializations}. For $m=1$ it gives linearized Reed--Solomon codes; for $t=1$, the rank-metric Reed--Muller codes of \cite{ACLN}. Taking $m=t=1$ gives Gabidulin codes. If $N=1$, the skew structure disappears and the codes become affine Cartesian codes on the grid
$\Lambda_1\times\cdots\times\Lambda_m\subseteq(\mathbb K^*)^m$, with the Hamming metric. Taking also $m=1$ gives Reed--Solomon codes with nonzero evaluation points.

Finally, we compare our construction with the linearized Reed--Muller codes of \cite{BC2025}, which are defined over finite fields. After a monomial change of variables, each of their codes is the image of a monomial code in one of our quotient algebras (Proposition~\ref{prop:BC-identification}); its support need not be a lower set and may come from Laurent monomials. This determines the dimension of all these codes. Conversely, our construction works over any finite Abelian Galois extension, including non-cyclic extensions over infinite fields, such as $\mathbb Q(\sqrt2,\sqrt3)/\mathbb Q$, and with arbitrary grids of norms. The families overlap, but neither contains the other in general.

When the support is a lower set, Theorem~\ref{thm:lower-set-distance} gives the exact minimum distance. This applies to the standard linearized Reed--Muller codes whose Frobenius powers have pairwise coprime orders with product $N$, and to the lattice-simplex codes associated with adapted lattice bases, including the almost-commutative family of \cite{BC2025}. For these codes we determine the exact minimum sum-rank distance in every degree, extending the parameter results of \cite[Theorems~3.1 and~3.3]{BC2025} beyond the range of degrees at most $q-2$, and we show that the improved bound of \cite[Remark~3.2]{BC2025} is sharp. Moreover, lower-set codes with hyperbolic support have larger dimension than the almost-commutative lattice-simplex codes with the same length and minimum distance.

\paragraph{Organization of the paper.}
Section~\ref{sec:preliminaries} fixes notation and recalls the necessary background on sum-rank metric codes, Abelian Galois extensions and skew polynomials, together with the codes that occur as special cases. Section~\ref{sec: rapp theorem} constructs the quotient representation of the sum-rank ambient space. Section~\ref{sec: Dickson and Moore} introduces the Dickson and Moore matrices and proves the rank identities for the sum-rank weight. Section~\ref{sezione codici} defines the total-degree, lower-set, and weighted sum-rank Reed--Muller-type codes and determines their dimensions and minimum distances. Section~\ref{sec: comparisons} gives the specializations and the comparison with Berardini and Caruso. The extension $\mathbb Q(\sqrt2,\sqrt3)/\mathbb Q$ serves as a running example.

\numberwithin{equation}{section}

\section{Preliminaries}\label{sec:preliminaries}

\subsection{Notation and sum-rank metric codes}\label{sec: notation preliminaries}

Throughout the paper, $\mathbb K\subseteq\mathbb L$ are fields and
$[\mathbb L:\mathbb K]=N<\infty$. We use
$\mathbb N=\{0,1,2,\ldots\}$. For a positive integer $a$, set $\Delta(a):=\{0,1,\ldots,a-1\}.$ If $\bm a=(a_1,\ldots,a_m)$ is a tuple of positive integers, write $\Delta(\bm a):=\Delta(a_1)\times\cdots\times\Delta(a_m)$.
For a multi-index $\bm i=(i_1,\ldots,i_m)\in\mathbb N^m$, we use
$|\bm i|:=i_1+\cdots+i_m$ and $\bm x^{\bm i}:=x_1^{i_1}\cdots x_m^{i_m}$.

We use the matrix form of the sum-rank metric. Let
$\bm\eta=(\eta_1,\ldots,\eta_t)$ satisfy
$1\leq\eta_j\leq N$, and set
\[
\eta:=\eta_1+\cdots+\eta_t,
\qquad
\mathcal M_{\bm\eta}
:=
\prod_{j=1}^t\mathbb K^{N\times\eta_j}.
\]
The integer $\eta$ is called the \emph{sum-rank length} of $\mathcal M_{\bm\eta}$.
For a tuple of matrices $A=\bigl(A^{(1)},\ldots,A^{(t)}\bigr)\in\mathcal M_{\bm\eta},$ its sum-rank weight is
\[
\srk(A):=\sum_{j=1}^t\rk_{\mathbb K}\bigl(A^{(j)}\bigr).
\]
The corresponding distance on $\mathcal M_{\bm\eta}$ is
\[
d_{\srk}(A,B)
:=\srk(A-B)
=\sum_{j=1}^t
\rk_{\mathbb K}\bigl(A^{(j)}-B^{(j)}\bigr).
\]
We only need the case $\mathcal M_{(N,\ldots,N)}=\bigl(\mathbb K^{N\times N}\bigr)^t$ of $t$ square blocks, which has sum-rank length $\eta=Nt$.

The usual vector description is obtained from the matrix model by choosing an
ordered $\mathbb K$-basis
$\mathcal B=(\beta_1,\ldots,\beta_N)$ of $\mathbb L$. For
$a\in\mathbb L$, let $[a]_{\mathcal B}\in\mathbb K^N$ denote its coordinate
column vector. If
$v^{(j)}=(v^{(j)}_1,\ldots,v^{(j)}_{\eta_j})\in
\mathbb L^{\eta_j}$, define
\[
\mathrm{Ext}_{\mathcal B}\bigl(v^{(j)}\bigr)
:=
\begin{bmatrix}
[v^{(j)}_1]_{\mathcal B}&\cdots&[v^{(j)}_{\eta_j}]_{\mathcal B}
\end{bmatrix}
\in\mathbb K^{N\times\eta_j}.
\]
Writing $v=\bigl(v^{(1)}\mid\cdots\mid v^{(t)}\bigr)\in\mathbb L^\eta$,
we obtain a $\mathbb K$-linear isomorphism
\[
\begin{aligned}
\mathrm{Ext}_{\mathcal B}^{\bm\eta}:
\mathbb L^\eta&\longrightarrow\mathcal M_{\bm\eta},\\
v&\longmapsto
\bigl(
\mathrm{Ext}_{\mathcal B}(v^{(1)}),\ldots,
\mathrm{Ext}_{\mathcal B}(v^{(t)})
\bigr).
\end{aligned}
\]
The sum-rank weight of $v$ is therefore
\[
\srk(v):=\srk\bigl(\mathrm{Ext}_{\mathcal B}^{\bm\eta}(v)\bigr)=\sum_{j=1}^t\rk_{\mathbb K}\bigl(\mathrm{Ext}_{\mathcal B}(v^{(j)})\bigr)=\sum_{j=1}^t\dim_{\mathbb K}\Span_{\mathbb K}\{v^{(j)}_1,\ldots,v^{(j)}_{\eta_j}\}.
\]
The associated distance on $\mathbb L^\eta$ is $d_{\srk}(u,v):=\srk(u-v).$ By construction, $\mathrm{Ext}_{\mathcal B}^{\bm\eta}$ is an isometry between the vector and matrix realizations of the sum-rank metric \cite{M-P2018,M-P2022}. Moreover, changing the basis $\mathcal B$ left-multiplies every matrix block by an invertible matrix over $\mathbb K$, so $\srk(v)$ does not depend on the chosen basis. 

If $t=1$, the weight is the rank of a single matrix and one recovers the
rank metric. If $\eta_1=\cdots=\eta_t=1$, each block has one column and its
rank is either zero or one; hence the sum-rank weight is the Hamming weight
\[
\mathrm{wt}_{\mathrm H}(v)=\#\{j\mid v_j\neq0\}.
\]
For a single vector block $v=(v_1,\ldots,v_\eta)\in\mathbb L^\eta$ we write
\[
  \rk_{\mathbb K}(v):=\rk_{\mathbb K}\bigl(\mathrm{Ext}_{\mathcal B}(v)\bigr)
  =\dim_{\mathbb K}\Span_{\mathbb K}\{v_1,\ldots,v_\eta\}
\]
and use $\rk(v)$ as an abbreviation for this rank weight.

A matrix sum-rank code is a nonzero $\mathbb K$-linear subspace
$\mathcal C\subseteq\mathcal M_{\bm\eta}$, with minimum distance
\[
  d_{\min}(\mathcal C)
  :=\min\{\srk(A)\mid 0\ne A\in\mathcal C\}.
\]
Since every block has $N$ rows and $\eta_j\le N$ columns, the
sum-rank Singleton bound takes the form
\[
  \dim_{\mathbb K}\mathcal C
  \le N\bigl(\eta-d_{\min}(\mathcal C)+1\bigr).
\]
For a nonzero $\mathbb L$-linear code $\mathcal C\subseteq\mathbb L^\eta$
of dimension $k=\dim_{\mathbb L}\mathcal C$, its matrix realization
\[
  \mathcal C_{\mathcal B}^{\mathrm{mat}}
  :=\mathrm{Ext}_{\mathcal B}^{\bm\eta}(\mathcal C)
  \subseteq\mathcal M_{\bm\eta}
\]
has dimension $Nk$ over $\mathbb K$ and the same minimum sum-rank
distance as $\mathcal C$. Consequently, the bound becomes
\[
  d_{\min}(\mathcal C)\le\eta-k+1.
\]
A matrix code attaining equality in the dimension bound is called
a \textit{maximum sum-rank distance code} (MSRD). An $\mathbb L$-linear code $\mathcal C$ is called MSRD if
$\mathcal C_{\mathcal B}^{\mathrm{mat}}$ is, that is, if
$d_{\min}(\mathcal C)=\eta-k+1$.
The corresponding terms in the rank and Hamming specializations
are \textit{maximum rank distance} (MRD) and \textit{maximum distance separable} (MDS), respectively.

\subsection{Classical polynomial evaluation codes}

We briefly recall the commutative model that motivates the construction. Let
$\mathbb F$ be a field, let $S_i\subseteq\mathbb F$ be a finite nonempty
set of cardinality $s_i$, and put
\[
S:=S_1\times\cdots\times S_m,
\qquad
g_i(x_i):=\prod_{a\in S_i}(x_i-a).
\]
The Cartesian evaluation map
\[
\begin{aligned}
\ev_S:\quad \mathbb F[x_1,\ldots,x_m]
&\longrightarrow \mathbb F^{S},\\
F&\longmapsto \bigl(F(P)\bigr)_{P\in S}
\end{aligned}
\]
has kernel $(g_1(x_1),\ldots,g_m(x_m))$. Consequently, every function on
$S$ has a unique polynomial representative in the span of the monomials
\[
\bm x^{\bm u},
\qquad
\bm u\in\Delta(\bm s),
\qquad
\bm s=(s_1,\ldots,s_m).
\]
The affine Cartesian code of order $r$ is the image under $\ev_S$ of the
span of the reduced monomials satisfying $|\bm u|\leq r$
\cite{GT2013,LRMV2014}. When $\mathbb F=\F_q$ and
$S_1=\cdots=S_m=\F_q$, this is the generalized Reed--Muller code
\cite{KLP68}; when $m=1$, it is a Reed--Solomon code \cite{RS60}. In each
case the dimension is obtained by counting the relevant reduced monomials,
and the minimum Hamming distance is determined by the least possible number
of nonzero values of such a polynomial on the grid. This last problem is the
setting of the Alon--F\"uredi and Schwartz--Zippel bounds
\cite{AF,S,Z,BCPS2018,CFS2017}.

\subsection{Abelian Galois extensions and skew group algebras}

From now on, let $\mathbb L/\mathbb K$ be a finite Abelian Galois extension
of degree $N$. Fix an internal direct-product decomposition
\[
G:=\Gal(\mathbb L/\mathbb K)
=\langle\sigma_1\rangle\times\cdots\times\langle\sigma_m\rangle
\cong
\frac{\mathbb Z}{n_1\mathbb Z}\times\cdots\times
\frac{\mathbb Z}{n_m\mathbb Z},
\]
where $n_i=\ord(\sigma_i)$ and $N=n_1\cdots n_m$. Set
$\bm n=(n_1,\ldots,n_m)$. For $\bm i=(i_1,\ldots,i_m)\in\mathbb Z^m$,
write $\bm\sigma^{\bm i}:=\sigma_1^{i_1}\cdots\sigma_m^{i_m}.$
The factors commute, and every element of $G$ is represented uniquely as
$\bm\sigma^{\bm i}$ with $\bm i\in\Delta(\bm n)$.

For each $i\in\{1,\ldots,m\}$, define $\mathbb E_i:=\mathbb L^{\langle\sigma_1,\ldots,\widehat{\sigma_i},\ldots,\sigma_m\rangle}.$
Then $[\mathbb E_i:\mathbb K]=n_i$, and the restriction of $\sigma_i$
generates $\Gal(\mathbb E_i/\mathbb K)$. We use the standard trace and norm
\[
\Tr_{\mathbb L/\mathbb K}(a)=\sum_{\tau\in G}\tau(a),
\qquad
\N_{\mathbb L/\mathbb K}(a)=\prod_{\tau\in G}\tau(a),
\]
and the analogous notation for intermediate extensions.

The skew group algebra $\mathbb L[G]$ is the left $\mathbb L$-vector
space
\[
\mathbb L[G]
=
\left\{\sum_{\tau\in G}a_\tau\tau\ \middle|\ a_\tau\in\mathbb L\right\},
\]
with multiplication determined by
\[
(a\tau)(b\gamma)=a\tau(b)(\tau\gamma),
\qquad
a,b\in\mathbb L,
\quad \tau,\gamma\in G.
\]
The natural action of $\mathbb L[G]$ on $\mathbb L$ defines a canonical
isomorphism of $\mathbb K$-algebras:
\[
\begin{aligned}
\iota:\quad \mathbb L[G]&\xrightarrow{\ \sim\ }\End_{\mathbb K}(\mathbb L),\\
\sum_{\tau\in G}a_\tau\tau
&\longmapsto
\left(z\longmapsto\sum_{\tau\in G}a_\tau\tau(z)\right).
\end{aligned}
\]

Fix an ordered $\mathbb K$-basis
$\mathcal B=(\beta_1,\ldots,\beta_N)$ of $\mathbb L$.
Through the preceding identification, each $P\in\mathbb L[G]$
is represented by the matrix
\[
  \rho_{\mathcal B}(P)
  :=\bigl([P(\beta_1)]_{\mathcal B}\ \cdots\ 
          [P(\beta_N)]_{\mathcal B}\bigr)
  \in\mathbb K^{N\times N}.
\]
Thus, $\rho_{\mathcal B}:\mathbb L[G]\to\mathbb K^{N\times N}$
is a $\mathbb K$-algebra isomorphism, and
\[
  \rk_{\mathbb K}(P)
  :=\dim_{\mathbb K}P(\mathbb L)
  =\rk_{\mathbb K}\bigl(\rho_{\mathcal B}(P)\bigr).
\]
Writing $\mathcal B^{(t)}=(\mathcal B,\ldots,\mathcal B)$,
we denote its componentwise extension by
\[
  \rho_{\mathcal B^{(t)}}:(\mathbb L[G])^t
  \longrightarrow(\mathbb K^{N\times N})^t,
  \qquad
  (P_1,\ldots,P_t)
  \longmapsto
  \bigl(\rho_{\mathcal B}(P_1),\ldots,\rho_{\mathcal B}(P_t)\bigr).
\]
We equip $(\mathbb L[G])^t$ with the sum-rank weight
\[
  \srk(P_1,\ldots,P_t)
  :=\sum_{j=1}^t\dim_{\mathbb K}P_j(\mathbb L)
  =\sum_{j=1}^t\rk_{\mathbb K}\bigl(\rho_{\mathcal B}(P_j)\bigr).
\]
Hence, $\rho_{\mathcal B^{(t)}}$ is an isometry for the corresponding
sum-rank distances. A change of basis simultaneously conjugates
the matrix blocks, so this weight is independent of $\mathcal B$.

\subsection{Multivariate skew polynomials and truncated norms}
Consider the multivariate skew polynomial ring $\Lx$ 
associated with the decomposition of $G$ fixed above. Its elements are finite sums
\[
F=\sum_{\bm i\in\mathbb N^m}F_{\bm i}\bm x^{\bm i},
\qquad F_{\bm i}\in\mathbb L,
\]
with coefficients written on the left. Addition is coefficientwise, and multiplication is determined by
\[
x_i a=\sigma_i(a)x_i,
\qquad
x_ix_j=x_jx_i,
\qquad
a\in\mathbb L.
\]
Equivalently,
\[
(a\bm x^{\bm i})(b\bm x^{\bm j})
=a\bm\sigma^{\bm i}(b)\bm x^{\bm i+\bm j}.
\]
These relations are compatible because the $\sigma_i$ commute. The support, total degree, and degree in $x_i$ are defined by
\[
\mathrm{Supp}(F):=\{\bm i\mid F_{\bm i}\neq0\},
\qquad
\deg F:=\max_{\bm i\in\mathrm{Supp}(F)}|\bm i|,
\qquad
\deg_{x_i}F:=\max_{\bm j\in\mathrm{Supp}(F)}j_i.
\]
As usual, by convention $\deg 0=-\infty$. For any admissible monomial order
$\prec$ on $\mathbb N^m$, we denote the leading monomial and leading term
by $\LM_\prec(F)$ and $\LT_\prec(F)$, respectively. If the order is clear,
the subscript is omitted. Since the twisting affects coefficients but not
exponents, leading monomials multiply as in the commutative polynomial ring.

For skew polynomials $P$ and $Q$, we say that $P$ is a \textbf{right divisor} of
$Q$ if $Q=AP$ for some skew polynomial $A$.

\begin{definition}
For $\bm\alpha=(\alpha_1,\ldots,\alpha_m)\in\mathbb E_1\times\cdots\times\mathbb E_m$ and $\bm i\in\mathbb N^m$, the \textbf{norm} and the $\bm i$-th \textbf{truncated norm} of $\bm\alpha$ are
\[
\N(\bm\alpha):=\prod_{\ell=1}^m\N_{\mathbb E_\ell/\mathbb K}(\alpha_\ell)=\prod_{\ell=1}^m\prod_{k=0}^{n_\ell-1}\sigma_\ell^k(\alpha_\ell)\in\mathbb K,
\qquad
\N^{\bm i}(\bm\alpha):=\prod_{\ell=1}^m\prod_{k=0}^{i_\ell-1}\sigma_\ell^k(\alpha_\ell)\in\mathbb L,
\]
where an empty product is $1$.
\end{definition}

Observe that the powers $x_i^{n_i}$ are central in $\Lx$. Indeed,
$\sigma_i^{n_i}=\id_{\mathbb L}$, and all variables commute. Therefore
$x_i^{n_i}-\lambda$ is central for every $\lambda\in\mathbb K$, and an
ideal generated by polynomials of this form is two-sided.

The basic operator evaluation used in the sequel is the following result from
\cite{ACLN}. Set
\[
I_{\bm n}:=(x_1^{n_1}-1,\ldots,x_m^{n_m}-1)\subseteq\Lx.
\]

\begin{theorem}[{\cite[Theorem~60]{ACLN}}]
\label{Theorem 60 ACLN}
The map
\[
\begin{aligned}
\psi:\quad \Lx&\longrightarrow\mathbb L[G],\\
\sum_{\bm i\in\mathbb N^m}b_{\bm i}\bm x^{\bm i}
&\longmapsto
\sum_{\bm i\in\mathbb N^m}b_{\bm i}\bm\sigma^{\bm i}
\end{aligned}
\]
is a surjective ring homomorphism with kernel $I_{\bm n}$.
\end{theorem}

\begin{corollary}
\label{Corollary Thm 60 ACLN}
The preceding map induces isomorphisms of $\mathbb K$-algebras
\[
\frac{\Lx}{I_{\bm n}}
\cong\mathbb L[G]
\cong\End_{\mathbb K}(\mathbb L)
\cong\mathbb K^{N\times N},
\]
where the last isomorphism is obtained after choosing a
$\mathbb K$-basis of $\mathbb L$.
\end{corollary}

\begin{example}\label{ex:running-example} 
We use the following example throughout the paper.
Consider
\[ 
\mathbb K=\mathbb Q, \qquad \mathbb L=\mathbb Q(\sqrt{2},\sqrt{3}), \qquad\text{with}\quad G=\Gal(\mathbb L/\mathbb K) = \langle \sigma_1,\sigma_2\rangle \cong \mathbb Z/2\mathbb Z\times\mathbb Z/2\mathbb Z,
\] 
where
\[ \sigma_1(\sqrt{2})=-\sqrt{2}, \quad \sigma_1(\sqrt{3})=\sqrt{3};\qquad
\sigma_2(\sqrt{2})=\sqrt{2}, \quad \sigma_2(\sqrt{3})=-\sqrt{3}.
\] 
Thus, $n_1=n_2=2$ and $N=[\mathbb L:\mathbb K]=4$.
The corresponding intermediate fields are 
\[ 
\mathbb{E}_1=\mathbb L^{\langle\sigma_2\rangle} =\mathbb Q(\sqrt{2}), \qquad \mathbb{E}_2=\mathbb L^{\langle\sigma_1\rangle} =\mathbb Q(\sqrt{3}).
\]
Let $\alpha_1=1+\sqrt{2}\in\mathbb{E}_1$ and $\alpha_2=1+\sqrt{3}\in\mathbb{E}_2.$
Their relative norms are
\[
\N_{\mathbb{E}_1/\mathbb K}(\alpha_1)=\prod_{h=0}^{1} \sigma_1^h(1+\sqrt{2})=(1+\sqrt{2})\sigma_1(1+\sqrt{2})=-1,\quad\N_{\mathbb{E}_2/\mathbb K}(\alpha_2)=\prod_{\ell=0}^{1} \sigma_2^\ell(1+\sqrt{3})=(1+\sqrt{3})\sigma_2(1+\sqrt{3})=-2.
\]
If we set $\bm{\alpha}=(\alpha_1,\alpha_2)$, then we can conclude
\[ 
\N(\bm{\alpha})=\prod_{\ell=1}^2
\big(\alpha_\ell\,\sigma_\ell({\alpha}_\ell)\big)=
\N_{\mathbb{E}_1/\mathbb K}(\alpha_1)\,\N_{\mathbb{E}_2/\mathbb K}(\alpha_2)=2.
\]
\end{example}

\subsection{Reference families in rank and sum-rank metrics}

We conclude by recording the special cases used later. Suppose first that
$\mathbb L/\mathbb K$ is cyclic of degree $N$, generated by
$\sigma$, and let $\mathcal B=(\beta_1,\ldots,\beta_N)$ be a
$\mathbb K$-basis of $\mathbb L$. For $0\leq r<N$, the operator form
of the Gabidulin code is
\[
\mathrm{Gab}_{\sigma}(r):=\Span_{\mathbb L}\{\mathrm{id},\sigma,\ldots,\sigma^r\}\subseteq\mathbb L[G].
\]
It has dimension $r+1$ and minimum rank distance $N-r$, and is therefore
MRD \cite{G85,ACLN}.

The linearized Reed--Solomon construction replaces the single rank block by several blocks. In the specialization needed in this paper, choose $\alpha_1,\ldots,\alpha_t\in\mathbb L^*$ with pairwise distinct norms in
$\mathbb K^*$, and set
\[
\N_{\mathbb L/\mathbb K}^{\,h}(\alpha)
:=\prod_{a=0}^{h-1}\sigma^a(\alpha).
\]
For $0\leq r<Nt$, the corresponding code is generated by
the rows
\[
\left(
\N_{\mathbb L/\mathbb K}^{\,h}(\alpha_j)
\sigma^h(\beta_k)
\right)_{
1\leq j\leq t,\ 1\leq k\leq N},
\qquad 0\leq h\leq r.
\]
It has dimension $r+1$ and minimum sum-rank distance $Nt-r$, hence it is
MSRD \cite{M-P2018}. For $t=1$, the nonzero truncated-norm factor only
rescales each generating row, and one recovers a Gabidulin code.

For the general Abelian decomposition fixed above, the rank-metric
Reed--Muller code of order $r$ introduced in \cite{ACLN} is the subspace
\[
\Span_{\mathbb L}
\left\{
\bm\sigma^{\bm u}
\ \middle|\
\bm u\in\Delta(\bm n),\ |\bm u|\leq r
\right\}
\subseteq\mathbb L[G],
\qquad
0\leq r\leq\sum_{i=1}^m(n_i-1),
\]
endowed with the rank of its action on $\mathbb L$. The construction of Berardini and Caruso \cite{BC2025} is recalled in Section~\ref{subsec:comparison-LRM}.

Corollary~\ref{Corollary Thm 60 ACLN} gives componentwise
isomorphisms
\[
(\mathbb L[G])^t
\cong
\left(\frac{\Lx}{I_{\bm n}}\right)^t,
\]
but there the $t$ blocks are represented by $t$ different polynomial classes. In the next section we construct a single quotient of $\Lx$ whose central relations encode all blocks at once.

\section{A skew-polynomial representation of the sum-rank ambient space}\label{sec: rapp theorem}
\begin{notation}\label{notation sezione 3}
Let $t=t_1\cdots t_m$, $N=n_1\cdots n_m$, and $\nu_i:=n_it_i$ for $i=1,\ldots,m$; put $\bm\nu:=(\nu_1,\ldots,\nu_m)$. Set
\[
J:=\{1,\dots,t_1\}\times\dots\times\{1,\dots,t_m\}.
\]
We order $J$ lexicographically and, through this order, identify it with $\{1,\ldots,t\}$. Accordingly, the $t$ blocks of $(\mathbb L[G])^t$, $(\mathbb K^{N\times N})^t$ and $\mathbb L^{Nt}$ are indexed by $\bm j\in J$, and the definitions of Section~\ref{sec: notation preliminaries} apply with $\sum_{j=1}^t$ replaced by $\sum_{\bm j\in J}$. Since the sum-rank weight is a sum over the blocks, it does not depend on the chosen order.
For every $i\in\{1,\dots,m\}$, fix nonzero elements $\alpha_{i,1},\ldots,\alpha_{i,t_i}\in\mathbb E_i$ and put $\lambda_{i,j}:=\N_{\mathbb{E}_i/\mathbb{K}}(\alpha_{i,j})\in\mathbb K^*$. Such a choice requires $t_i\le\#\N_{\mathbb E_i/\mathbb K}(\mathbb E_i^*)$; for instance, $t_i\le q-1$ if $\mathbb K=\F_q$. We assume that
\begin{equation}
\lambda_{i,j}\neq\lambda_{i,k}\quad\forall j\neq k.
\end{equation}
We write
\[\Lambda_i:=\{\lambda_{i,1},\dots,\lambda_{i,t_i}\}\quad\text{for }i\in\{1,\dots,m\},\qquad\Lambda:=(\Lambda_1,\ldots,\Lambda_m).\]
Moreover, for $\bm j=(j_1,\ldots,j_m)\in J$, we write $\bm{\alpha}^{(\bm j)}:=(\alpha_{1,j_1},\ldots,\alpha_{m,j_m}).$
For every $i\in\{1,\ldots,m\}$, we set
\[
H_i(x_i):=\prod_{j_i=1}^{t_i}(x_i^{n_i}-\lambda_{i, j_i}).
\]
We denote by $H_\Lambda:=(H_1(x_1),\dots,H_m(x_m))$ the ideal of $\mathbb L[x_1,\ldots,x_m;\sigma_1,\ldots,\sigma_m]$ generated by these polynomials.

Furthermore, for every $\bm j=(j_1,\ldots,j_m)\in J$, let $I_{\bm j}:= (x_1^{n_1}-\lambda_{1,j_1},\ldots, x_m^{n_m}-\lambda_{m,j_m}).$
\end{notation}
\begin{remark}\label{remark Polinomi centrali e comaximal}
For every $i\in\{1,\ldots,m\}$ and every $j_i\in\{1,\ldots,t_i\}$, the polynomial $x_i^{n_i}-\lambda_{i,j_i}$ is central. Indeed, for every $a\in\mathbb L$
\[
x_i^{n_i}a=\sigma_i^{n_i}(a)x_i^{n_i}=ax_i^{n_i}, 
\] 
since $\sigma_i^{n_i}=\id_{\mathbb{L}}$. Moreover, all the variables $x_1,\ldots,x_m$ commute with one another and $\lambda_{i,j_i}\in\mathbb{K}=\mathbb{L}^{G}$. Consequently, the polynomials $H_i(x_i)$ are central and all $I_{\bm{j}}$,
as well as $H_{\Lambda}$, are two-sided ideals of $\mathbb L[x_1,\dots,x_m; \sigma_1,\dots,\sigma_m]$.
Moreover, the ideals $I_{\bm{j}}$, with $\bm{j}\in J$, are pairwise comaximal. Indeed, let $\bm{j}\neq\bm{k}$. Then there exists $i\in\{1,\dots,m\}$ such that $j_i\neq k_i$. In our setting this means $\lambda_{i,j_i}\neq\lambda_{i,k_i}$, and
\[ 
1= \bigl(\lambda_{i,k_i}-\lambda_{i,j_i}\bigr)^{-1} \left[ \bigl(x_i^{n_i}-\lambda_{i,j_i}\bigr) - \bigl(x_i^{n_i}-\lambda_{i,k_i}\bigr) \right] \in I_{\bm{j}}+I_{\bm{k}}. 
\] 
Thus, 
\[ 
I_{\bm{j}}+I_{\bm{k}}=\mathbb L[x_1,\dots,x_m; \sigma_1,\dots,\sigma_m] \qquad \forall\bm{j}\neq\bm{k}.
\] 
\end{remark}
\begin{remark}\label{rem:universal-grobner-basis}
Put $y_i:=x_i^{n_i}$ and $h_i(y):=\prod_{j=1}^{t_i}(y-\lambda_{i,j})\in\mathbb K[y]$, so that $H_i(x_i)=h_i(y_i)$. The elements $y_1,\ldots,y_m$ are central, so they generate a commutative polynomial ring $\mathbb L[y_1,\ldots,y_m]$ inside $\Lx$, and $\Lx$ is a free left module over this ring with basis $\{\bm x^{\bm a}\mid\bm a\in\Delta(\bm n)\}$. Since every $H_i$ is central,
\[
H_\Lambda=\sum_{i=1}^m\Lx\,H_i
=\bigoplus_{\bm a\in\Delta(\bm n)}\Bigl(\sum_{i=1}^m\mathbb L[y_1,\ldots,y_m]\,h_i(y_i)\Bigr)\bm x^{\bm a}.
\]
The classes of $\bm y^{\bm b}$, $\bm b\in\Delta(t_1,\ldots,t_m)$, form an $\mathbb L$-basis of $\mathbb L[y_1,\ldots,y_m]/(h_1(y_1),\ldots,h_m(y_m))$, and $\bm y^{\bm b}\bm x^{\bm a}=\bm x^{\bm a+\bm n\bm b}$ with $\bm n\bm b:=(n_1b_1,\ldots,n_mb_m)$. Every $\bm u\in\Delta(\bm\nu)$ is uniquely of the form $\bm a+\bm n\bm b$ with $\bm a\in\Delta(\bm n)$ and $\bm b\in\Delta(t_1,\ldots,t_m)$. Hence, the standard monomials
\[
\mathcal B_{\mathrm{mon}}:=\left\{\bm x^{\bm a}\;\middle|\;\bm a\in\Delta(\bm\nu)\right\}
\]
form an $\mathbb L$-basis of $\Lx/H_\Lambda$, and $\dim_{\mathbb L}\Lx/H_\Lambda=Nt$. Equivalently, $\mathcal G_\Lambda:=\{H_1(x_1),\ldots,H_m(x_m)\}$, with $\LM_\prec(H_i)=x_i^{\nu_i}$, is the reduced Gröbner basis of $H_\Lambda$ for every monomial order $\prec$. Every residue class has a unique representative
\[
F_{\mathrm{red}}=\sum_{\bm a\in\Delta(\bm\nu)}c_{\bm a}\bm x^{\bm a},
\qquad c_{\bm a}\in \mathbb L,
\]
the reduced representative of the class; when no confusion can arise, we identify the class with $F_{\mathrm{red}}$. The same argument with $t_1=\cdots=t_m=1$ shows that, for every $\bm j\in J$, the classes of $\bm x^{\bm a}$, $\bm a\in\Delta(\bm n)$, form an $\mathbb L$-basis of $\Lx/I_{\bm j}$.
\end{remark}
\begin{lemma}\label{lem:intersection-ideals}
With the notation introduced in Notation \ref{notation sezione 3}, we have
\[
\bigcap_{\bm j\in J}
I_{\bm{j}}
=
\big(H_1(x_1),\ldots,H_m(x_m)\big).
\]
\end{lemma}
\begin{proof}
We first prove that 
\[ (H_1(x_1),\ldots,H_m(x_m)) \subseteq \bigcap_{\bm j\in J} I_{\bm j}.
\] 
Indeed, fix an arbitrary $\bm j=(j_1,\ldots,j_m)\in J$. For every $i\in\{1,\ldots,m\}$, we have 
\[
H_i(x_i)=\left(\prod_{\substack{s=1\\s\neq j_i}}^{t_i}(x_i^{n_i}-\lambda_{i,s})\right)
(x_i^{n_i}-\lambda_{i, j_i})
\]
and therefore $H_i(x_i)\in I_{\bm j}$, since the factor $x_i^{n_i}-\lambda_{i, j_i}$ belongs to $I_{\bm j}$. Hence, 
\[ 
(H_1(x_1),\ldots,H_m(x_m))\subseteq I_{\bm j}.
\] 
Since this holds for every $\bm j\in J$, the desired inclusion follows.

Observe that the dimensions of the corresponding quotients as $\mathbb L$-vector spaces are finite. Hence, it suffices to show the equality between them in order to prove the reverse inclusion, i.e.
\begin{equation}\label{uguaglianza tra dimensioni}
    \dim_{\mathbb{L}}\frac{\mathbb{L}[x_1,\dots,x_m;\sigma_1,\dots,\sigma_m]}{\bigcap_{\bm j\in J}
I_{\bm{j}}}=\dim_{\mathbb{L}}\frac{\mathbb{L}[x_1,\dots,x_m;\sigma_1,\dots,\sigma_m]}{\Big(H_1(x_1),\dots,H_m(x_m)\Big)}.
\end{equation}
As stated in Remark \ref{remark Polinomi centrali e comaximal}, all $I_{\bm{j}}$ are pairwise comaximal two-sided ideals. As a consequence, the Chinese Remainder Theorem gives
\begin{equation}\label{crt}
\frac{\mathbb{L}[x_1,\dots,x_m;\sigma_1,\dots,\sigma_m]}{\bigcap_{\bm j\in J}
I_{\bm{j}}}\cong\prod_{\bm{j}\in J}\frac{\mathbb{L}[x_1,\dots,x_m;\sigma_1,\dots,\sigma_m]}{
I_{\bm{j}}}.
\end{equation}
By Remark~\ref{rem:universal-grobner-basis}, for every $\bm j\in J$ the residue classes of the monomials 
\[ 
x_1^{a_1}\cdots x_m^{a_m}, \qquad 0\leq a_i<n_i, 
\] 
form an $\mathbb L$-basis of $\mathbb{L}[x_1,\dots,x_m;\sigma_1,\dots,\sigma_m]/I_{\bm j}$. 
Therefore, , 
\[ 
\dim_{\mathbb L}\frac{\mathbb{L}[x_1,\dots,x_m;\sigma_1,\dots,\sigma_m]}{I_{\bm j}} = \prod_{i=1}^m n_i = N. 
\] 
Since $\#J=\prod_{i=1}^m t_i=t$, it follows that 
\[ 
\dim_{\mathbb L} \frac{\mathbb{L}[x_1,\dots,x_m;\sigma_1,\dots,\sigma_m]}{\displaystyle\bigcap_{\bm j\in J}I_{\bm j}} = \dim_{\mathbb L}\prod_{\bm{j}\in J}\frac{\mathbb{L}[x_1,\dots,x_m;\sigma_1,\dots,\sigma_m]}{
I_{\bm{j}}}= Nt.
\] 
On the other hand, by Remark~\ref{rem:universal-grobner-basis}, the residue classes of the monomials 
\[ 
x_1^{a_1}\cdots x_m^{a_m}, \qquad 0\leq a_i<n_it_i, \] 
form an $\mathbb L$-basis of 
$\mathbb{L}[x_1,\dots,x_m;\sigma_1,\dots,\sigma_m]/\big(H_1(x_1),\ldots,H_m(x_m)\big).
$ 
Consequently, 
\[ 
\dim_{\mathbb L} \frac{\mathbb{L}[x_1,\dots,x_m;\sigma_1,\dots,\sigma_m]}{(H_1(x_1),\ldots,H_m(x_m))} = \prod_{i=1}^m n_it_i = Nt.
\] 
This proves the claim in \eqref{uguaglianza tra dimensioni} and consequently the statement of this lemma.
\end{proof}
The next step is to introduce a \textbf{representation theorem} allowing us to view every element of
$(\mathbb L[G])^{t}$ as the class of an element in a suitable quotient of the multivariate skew polynomial ring
$\mathbb L[x_1,\ldots,x_m;\sigma_1,\ldots,\sigma_m]$.
More precisely, we will construct a surjective ring homomorphism
\[
\bm{\Phi}:\ \mathbb{L}[x_1,\ldots,x_m;\sigma_1,\ldots,\sigma_m]\longrightarrow (\mathbb{L}[G])^{t}
\]
and describe its kernel explicitly, thus obtaining an isomorphism
\[
(\mathbb{L}[G])^{t}\ \cong\ \frac{\mathbb{L}[x_1,\ldots,x_m;\sigma_1,\ldots,\sigma_m]}{\ker(\bm{\Phi})}.
\]
To construct the desired isomorphism, we first introduce the auxiliary maps involved in its definition and prove the preliminary results needed for the construction.
\begin{lemma}\label{lem:phi-ring-homomorphism}
For $\bm j\in J$, let $\phi_{\bm j}\colon\Lx\to\Lx$ be the substitution
\[
\phi_{\bm j}\Bigl(\sum_{\bm i}c_{\bm i}\bm x^{\bm i}\Bigr):=\sum_{\bm i}c_{\bm i}(\alpha_{1,j_1}x_1)^{i_1}\cdots(\alpha_{m,j_m}x_m)^{i_m},
\]
and let $\phi:=(\phi_{\bm j})_{\bm j\in J}\colon\Lx\to\Lx^t$. Then every $\phi_{\bm j}$ is a ring automorphism, and $\phi$ is a ring homomorphism.
\end{lemma}
\begin{proof}
Fix $\bm j\in J$ and put $y_i:=\alpha_{i,j_i}x_i$. For $a\in\mathbb L$ we have $y_ia=\alpha_{i,j_i}\sigma_i(a)x_i=\sigma_i(a)y_i$. For $i\neq\ell$, since $\sigma_i$ fixes $\mathbb E_\ell$ and $\sigma_\ell$ fixes $\mathbb E_i$,
\[
y_iy_\ell=\alpha_{i,j_i}\sigma_i(\alpha_{\ell,j_\ell})x_ix_\ell=\alpha_{i,j_i}\alpha_{\ell,j_\ell}x_ix_\ell=\alpha_{\ell,j_\ell}\sigma_\ell(\alpha_{i,j_i})x_\ell x_i=y_\ell y_i.
\]
Thus, $y_1,\ldots,y_m$ satisfy the defining relations of $\Lx$, and $\phi_{\bm j}$ is a ring endomorphism. By the same argument, the substitution $x_i\mapsto\alpha_{i,j_i}^{-1}x_i$ is an endomorphism; it is inverse to $\phi_{\bm j}$. Since the ring operations of $\Lx^t$ are componentwise, $\phi$ is a ring homomorphism.
\end{proof}
For $\bm j\in J$ we denote by $\pi_{\bm j}\colon\Lx^t\to\Lx$ the projection onto the $\bm j$-th component, so that $\pi_{\bm j}\circ\phi=\phi_{\bm j}$.
\begin{theorem}[\textbf{Representation Theorem}]\label{representation theorem}
Let $\psi$ be the map of Theorem~\ref{Theorem 60 ACLN}, and define
\[
\bm\Phi\colon\Lx\longrightarrow(\mathbb L[G])^t,\qquad
p\longmapsto\bigl(\psi(\phi_{\bm j}(p))\bigr)_{\bm j\in J}.
\]
Then $\bm\Phi$ is a surjective ring homomorphism with $\ker\bm\Phi=H_\Lambda$. In particular,
\[
(\mathbb L[G])^t\cong\frac{\Lx}{H_\Lambda}.
\]
\end{theorem}
\begin{proof}
By Lemma~\ref{lem:phi-ring-homomorphism} and Theorem~\ref{Theorem 60 ACLN}, $\bm\Phi$ is a ring homomorphism, and $p\in\ker\bm\Phi$ if and only if $\phi_{\bm j}(p)\in I_{\bm n}$ for every $\bm j\in J$. Hence
\[
\ker\bm\Phi=\bigcap_{\bm j\in J}\phi_{\bm j}^{-1}(I_{\bm n}).
\]
Fix $\bm j\in J$. Since $\phi_{\bm j}$ is an automorphism, $\phi_{\bm j}^{-1}(I_{\bm n})$ is the two-sided ideal generated by $\phi_{\bm j}^{-1}(x_i^{n_i}-1)=(\alpha_{i,j_i}^{-1}x_i)^{n_i}-1$, $i=1,\ldots,m$. As $\alpha_{i,j_i}\in\mathbb E_i$ and $\sigma_i|_{\mathbb E_i}$ generates $\Gal(\mathbb E_i/\mathbb K)$,
\[
(\alpha_{i,j_i}^{-1}x_i)^{n_i}=\N_{\mathbb E_i/\mathbb K}(\alpha_{i,j_i}^{-1})\,x_i^{n_i}=\lambda_{i,j_i}^{-1}x_i^{n_i}.
\]
Therefore,  $\phi_{\bm j}^{-1}(x_i^{n_i}-1)=\lambda_{i,j_i}^{-1}(x_i^{n_i}-\lambda_{i,j_i})$ and $\phi_{\bm j}^{-1}(I_{\bm n})=I_{\bm j}$. By Lemma~\ref{lem:intersection-ideals}, $\ker\bm\Phi=H_\Lambda$. Consequently, $\bm\Phi$ induces an injective $\mathbb L$-linear map $\Lx/H_\Lambda\to(\mathbb L[G])^t$. By Remark~\ref{rem:universal-grobner-basis}, both spaces have $\mathbb L$-dimension $Nt$, so this map is bijective and $\bm\Phi$ is surjective.
\end{proof}
\begin{remark}
Since $H_{\Lambda}=\ker\bm\Phi$, the map $\bm\Phi$ induces an
isomorphism
\[
\overline{\bm\Phi}\colon
\frac{
\mathbb L[x_1,\ldots,x_m;\sigma_1,\ldots,\sigma_m]
}{
H_{\Lambda}
}
\longrightarrow
(\mathbb L[G])^t,
\qquad
\overline{\bm\Phi}(F+H_{\Lambda})=\bm\Phi(F).
\]
By a slight abuse of notation, we will denote the induced map
$\overline{\bm\Phi}$ again by $\bm\Phi$. Moreover, in view of Remark~\ref{rem:universal-grobner-basis}, whenever we write $\bm\Phi(F)$ for an element of the quotient, $F$ is understood to be its reduced representative.
\end{remark}
\begin{example}[Continuation of Example~\ref{ex:running-example}]\label{ex:running-example,representation}
Let $t_1=t_2=2$, so that $t=t_1t_2=4$, and choose
\[
(\alpha_{1,1},\alpha_{1,2})=(1,1+\sqrt{2})\in(\mathbb E_1^*)^2, \qquad (\alpha_{2,1},\alpha_{2,2})=(1,1+\sqrt{3})\in(\mathbb E_2^*)^2.
\] 
Their relative norms are 
\[ 
\lambda_{1,1}=\N_{\mathbb{E}_1/\mathbb K}(\alpha_{1,1})=1, \qquad \lambda_{1,2}=\N_{\mathbb{E}_1/\mathbb K}(\alpha_{1,2})=-1, 
\] 
and 
\[ 
\lambda_{2,1}=\N_{\mathbb{E}_2/\mathbb K}(\alpha_{2,1})=1, \qquad \lambda_{2,2}=\N_{\mathbb{E}_2/\mathbb K}(\alpha_{2,2})=-2.
\] 
Note that the distinct-norm condition is satisfied. Then the associated polynomials are 
\[ H_1(x_1) = (x_1^2-1)(x_1^2+1) = x_1^4-1 \qquad \text{and}\qquad H_2(x_2) = (x_2^2-1)(x_2^2+2) = x_2^4+x_2^2-2. 
\]
Hence, the quotient algebra associated with this example is 
\[ 
\mathcal R := \frac{ \mathbb L[x_1,x_2;\sigma_1,\sigma_2] }{ (x_1^4-1,\;x_2^4+x_2^2-2) }.
\]
Moreover, 
\[ 
\dim_{\mathbb L}\mathcal R = (n_1t_1)(n_2t_2) = 16,
\]
and a natural ordered basis of $\mathcal R$ is $\mathcal B_{\mathrm{mon}} = \left( x_1^{i_1}x_2^{i_2} \;\middle|\; 0\leq i_1,i_2<4 \right),$ where the monomials are ordered lexicographically.

By the Representation Theorem, 
\[ 
\mathcal R \cong (\mathbb L[G])^4.
\] 
Thus, this example yields a sum-rank ambient space consisting of four blocks, each corresponding to a copy of the skew group algebra $\mathbb L[G]$.
\end{example}

\section{Dickson and Moore matrices over Abelian Galois extensions}\label{sec: Dickson and Moore}
In this section we show that the sum-rank weight of $\bm\Phi(F)$ equals the rank over $\mathbb L$ of the matrix of right multiplication by $F$ (a Dickson matrix), and also the rank of a generalized Moore matrix. This is the main tool for the distance bounds of Section~\ref{sezione codici}.

\subsection{Dickson matrices associated with right multiplication}
\begin{definition}\label{def:dickson-matrix}
Let $F\in\Lx/H_\Lambda$. Right multiplication by $F$ is the $\mathbb L$-linear map
\[
\mu_F\colon\Lx/H_\Lambda\longrightarrow\Lx/H_\Lambda,\qquad g\longmapsto gF.
\]
Order $\Delta(\bm\nu)$ lexicographically, and let
\[
\mathcal B_{\mathrm{mon}}=\bigl(\bm x^{\bm i}\bigr)_{\bm i\in\Delta(\bm\nu)}
\]
be the corresponding ordered basis. The \textbf{Dickson matrix} of $F$ is
\[
D_{\mathcal B_{\mathrm{mon}}}(F):=\Mat_{\mathcal B_{\mathrm{mon}}}(\mu_F)\in\mathbb L^{Nt\times Nt};
\]
its column indexed by $\bm i$ is the coordinate vector of $\bm x^{\bm i}F$ with respect to $\mathcal B_{\mathrm{mon}}$.
\end{definition}
\begin{remark}\label{rk=dimrispcolonideal}
By definition, $\rk_{\mathbb L}D_{\mathcal B_{\mathrm{mon}}}(F)=\dim_{\mathbb L}\mathrm{Im}(\mu_F)$, which does not depend on the chosen basis. Let
\[
(H_\Lambda:F):=\{g\in\Lx\mid gF\in H_\Lambda\}
\]
be the left colon ideal; it depends only on the class of $F$. Then $\ker(\mu_F)=(H_\Lambda:F)/H_\Lambda$, and therefore
\[
\rk_{\mathbb L}D_{\mathcal B_{\mathrm{mon}}}(F)=Nt-\dim_{\mathbb L}\frac{(H_\Lambda:F)}{H_\Lambda}=\dim_{\mathbb L}\frac{\Lx}{(H_\Lambda:F)}.
\]
\end{remark}

\begin{example}[Continuation of Example~\ref{ex:running-example}]\label{ex:running-dickson}
Recall that
\[
\mathcal R
=
\frac{\mathbb L[x_1,x_2;\sigma_1,\sigma_2]}
{(x_1^4-1,\;x_2^4+x_2^2-2)}
\]
and consider the element
\[
F=1+x_1+x_2+x_1x_2=(1+x_1)(1+x_2).
\]
We choose the lexicographically ordered monomial basis
\[
\mathcal B_{\mathrm{mon}}
=(1,x_2,x_2^2,x_2^3, x_1,x_1x_2,x_1x_2^2,x_1x_2^3, x_1^2,x_1^2x_2,x_1^2x_2^2,x_1^2x_2^3, x_1^3,x_1^3x_2,x_1^3x_2^2,x_1^3x_2^3).
\]
More precisely,  $(i_1,i_2)\lex(k_1,k_2)$ if either 
\[ i_1<k_1\qquad\text{or}\quad i_1=k_1 \quad\text{and}\quad i_2<k_2.\] 
Let 
\[ 
\mu_F:\mathcal R\longrightarrow\mathcal R, \qquad g\longmapsto gF, 
\] 
be the right-multiplication map associated with $F$. 
The map $\mu_F$ can be rewritten as the composition 
\begin{equation}\label{composizione moltiplicazione esempio dickson}
\mu_F = \mu_{1+x_2}\circ\mu_{1+x_1}. 
\end{equation} 
Indeed,
\[ \bigl(\mu_{1+x_2}\circ\mu_{1+x_1}\bigr)(g) = \mu_{1+x_2}\bigl(g\cdot(1+x_1)\bigr) = g\cdot(1+x_1)(1+x_2) = gF. 
\]
We now describe the matrices of these two multiplication maps. For the first variable, consider the auxiliary quotient 
\[ 
\mathcal R_1 = \frac{\mathbb L[x_1;\sigma_1]}{(x_1^4-1)} 
\] 
with ordered $\mathbb L$-basis $\mathcal B_1=(1,x_1,x_1^2,x_1^3)$.
Let $C_1$ be the matrix, with respect to $\mathcal B_1$, of right multiplication by $x_1$. Since 
\[ 
1\longmapsto x_1,\qquad x_1\longmapsto x_1^2,\qquad x_1^2\longmapsto x_1^3,\qquad x_1^3\longmapsto x_1^4=1, 
\]
we obtain 
\[ C_1= \begin{pmatrix} 
0&0&0&1\\ 
1&0&0&0\\ 
0&1&0&0\\ 
0&0&1&0 
\end{pmatrix}. 
\]
Consequently, the matrix of right multiplication by $1+x_1$ on $\mathcal R_1$ is 
\[ A_1:=I_4+C_1 = \begin{pmatrix} 
1&0&0&1\\ 
1&1&0&0\\ 
0&1&1&0\\ 
0&0&1&1 
\end{pmatrix}. 
\]
Similarly, consider \[ \mathcal R_2 = \frac{\mathbb L[x_2;\sigma_2]} {(x_2^4+x_2^2-2)} \]
with ordered basis $\mathcal B_2=(1,x_2,x_2^2,x_2^3)$. Let $C_2$ be the matrix of right multiplication by $x_2$. We have
\[ 
1\longmapsto x_2,\qquad x_2\longmapsto x_2^2,\qquad x_2^2\longmapsto x_2^3,\qquad x_2^3\longmapsto x_2^4=2-x_2^2. \] 
Using the same calculation as above, we obtain the following matrix of right multiplication by $1+x_2$
\[ A_2:=I_4+C_2 = \begin{pmatrix} 
1&0&0&2\\ 
1&1&0&0\\ 
0&1&1&-1\\
0&0&1&1 \end{pmatrix}. 
\]
We now move to the complete quotient $\mathcal R$. At the level of ordered $\mathbb L$-bases, we identify 
$\mathcal B_{\mathrm{mon}}$
with the product basis $\mathcal B_1\otimes\mathcal B_2, $ through the correspondence 
\[ x_1^{i_1}x_2^{i_2} \longleftrightarrow x_1^{i_1}\otimes x_2^{i_2}. 
\] 
Hence, 
\[ \Mat_{\mathcal B_{\mathrm{mon}}} (\mu_{1+x_1}) = A_1\otimes I_4\qquad\text{and}\qquad \Mat_{\mathcal B_{\mathrm{mon}}} (\mu_{1+x_2}) = I_4\otimes A_2. \]
By \eqref{composizione moltiplicazione esempio dickson} it follows that the Dickson matrix of $F$ is obtained by the ordinary matrix product of the matrices of the two maps: 
\[D_{\mathcal B_{\mathrm{mon}}}(F)= \Mat_{\mathcal B_{\mathrm{mon}}} (\mu_{1+x_2}) \Mat_{\mathcal B_{\mathrm{mon}}} (\mu_{1+x_1})= (I_4\otimes A_2)(A_1\otimes I_4).
\]
Using the properties of the Kronecker product, we obtain 
\[
D_{\mathcal B_{\mathrm{mon}}}(F) =(I_4\otimes A_2)(A_1\otimes I_4)= (I_4A_1)\otimes(A_2I_4)= A_1\otimes A_2.
\] 
Thus, we may write the Dickson matrix in block form as 
\[ D_{\mathcal B_{\mathrm{mon}}}(F) = A_1\otimes A_2 = \begin{pmatrix}
A_2&0&0&A_2\\ 
A_2&A_2&0&0\\ 
0&A_2&A_2&0\\
0&0&A_2&A_2 
\end{pmatrix} \in\mathbb L^{16\times16}.
\] 
Each displayed entry is a $4\times4$ block. 
As a consistency check, the first column of $D_{\mathcal B_{\mathrm{mon}}}(F)$ is the coordinate vector of \[ 1\cdot F=1+x_1+x_2+x_1x_2 \] with respect to $\mathcal B_{\mathrm{mon}}$, namely $( 
1,1, 0, 0, 1, 1, 0, 0, 0, 0, 0, 0, 0, 0, 0, 0)^{\mathsf T}.$
Finally, a direct computation shows that $\rk_{\mathbb L}(A_1)=\rk_{\mathbb L}(A_2)=3.$
In conclusion,
\[
\rk_{\mathbb L}
\bigl(D_{\mathcal B_{\mathrm{mon}}}(F)\bigr)
=
\rk_{\mathbb L}(A_1)
\rk_{\mathbb L}(A_2)
=9.
\]
\end{example}

\subsection{Generalized Moore matrices and the sum-rank weight}
\begin{notation}\label{notazione moore}
Let $v\in\mathbb L^{Nt}$. Using the order on $J$ fixed in Notation \ref{notation sezione 3}, we write $v$ as the ordered tuple of blocks $v=\bigl(v^{(\bm j)}\bigr)_{\bm j\in J},$ where for each multi-index $\bm{j} \in J$, the corresponding block is a vector in $\mathbb{L}^N$ defined as:
\[
v^{(\bm{j})} = (v_{\bm{j},1}, v_{\bm{j},2}, \dots, v_{\bm{j},N}).
\]
The rows of the generalized Moore matrix will be indexed by
$\bm i\in\Delta(\bm\nu)$, while its columns will be indexed by $(\bm j,k)\in J\times\{1,\ldots,N\},$ using the lexicographic order on both index sets.
\end{notation}
We now define a Moore matrix adapted to the decomposition $G=\langle\sigma_1\rangle\times\cdots\times\langle\sigma_m\rangle$ and to the family $\bm\alpha=(\alpha_{i,j})$ fixed in Notation~\ref{notation sezione 3}.
\begin{definition}\label{def moore matrix}
Let $v \in \mathbb{L}^{Nt}$. The $(G, \bm{\alpha})$-Moore matrix of $v$ is the matrix $M_{G,\bm\alpha}(v)\in\mathbb{L}^{Nt\times Nt}$ with entries
\[
M_{G,\bm\alpha}(v)_{\bm i,(\bm j,k)} := \N^{\bm i}(\bm\alpha^{(\bm j)})\bm\sigma^{\bm i}(v_{\bm j,k})
\]
for $\bm{i} \in\Delta(\bm\nu)$ and $(\bm{j},k) \in J \times \{1, \dots, N\}$.
\end{definition}

\begin{lemma}\label{lem:block-moore-independence}
Let $w_1,\ldots,w_r\in\mathbb L$ be linearly independent over $\mathbb K$, let $\bm w=(w_1,\ldots,w_r,0,\ldots,0)\in\mathbb L^N$, and let $\bm\alpha\in\mathbb E_1^*\times\cdots\times\mathbb E_m^*$. Consider the matrix $B\in\mathbb L^{Nt\times N}$ with entries
\[
B_{\bm i,k} := \N^{\bm i}(\bm\alpha)\bm\sigma^{\bm i}(w_k)\quad \text{for}\;\bm i\in\Delta(\bm\nu)\; \text{and}\; 1 \le k \le N. 
\]
Then, the $r$ nonzero columns of $B$ are linearly independent over $\mathbb{L}$, i.e., $\rk(B) = r$.
\end{lemma}
\begin{proof}
By contradiction, assume that the $r$ nonzero columns of $B$ are not linearly independent over $\mathbb{L}$.
Then
\begin{equation}\label{lin dip block}
\N^{\bm{i}}(\bm\alpha)\sum_{k=1}^rd_k\bm{\sigma}^{\bm i}(w_k)=0\qquad \forall \bm{i}\in\Delta(\bm\nu),\;\text{for some $d_1,\ldots,d_r\in\mathbb L$, not all zero}.
\end{equation}
Without loss of generality, we can assume $d_1\neq0$ and, up to a scalar multiple, $\Tr_{\mathbb{L}/\mathbb{K}}(d_1)\neq0$, due to the nondegeneracy of the trace pairing for finite separable
extensions.

Since $\N^{\bm{i}}(\bm{\alpha})$ is nonzero for every $\bm{i}\in\Delta(\bm\nu)$, we can rewrite \eqref{lin dip block} as follows
\begin{equation}\label{relation over I}
\sum_{k=1}^r\bm\sigma^{\bm{-i}}(d_k)\,w_k=0 \qquad \forall \bm{i}\in\Delta(\bm\nu)   ,
\end{equation}
after applying $\bm{\sigma}^{\bm{-i}}$. If we restrict the relations in \eqref{relation over I} to $\Delta(\bm n) = \Delta(n_1)\times\dots\times\Delta(n_m)\subseteq\Delta(\bm\nu)$,  and we sum them over all $\bm i\in\Delta(\bm n)$, we obtain
\begin{equation*}
\sum_{\bm{i}\in\Delta(\bm n)}\sum_{k=1}^r\bm{\sigma}^{\bm{-i}}(d_k)\,w_k=0   ,
\end{equation*}
which is equivalent to
\begin{equation}\label{recognize trace}
\sum_{k=1}^r\sum_{\bm{i}\in\Delta(\bm n)}\bm\sigma^{\bm{-i}}(d_k)\,w_k=0.  
\end{equation}
By the definition of the trace, this reads
\begin{equation*}
\sum_{k=1}^r\Tr_{\mathbb{L}/\mathbb{K}}(d_k)\,w_k=0  .
\end{equation*}
By hypothesis, the $w_k$ are all linearly independent over $\mathbb{K}$. This means that
\[
\Tr_{\mathbb{L}/\mathbb{K}}(d_k)=0\qquad\forall k=1,\dots,r,
\]
but it leads to a contradiction since we assumed $\Tr_{\mathbb{L}/\mathbb{K}}(d_1)\neq0$.
\end{proof}

\begin{lemma}\label{lem:moore-periodicity}
Let $\bm j=(j_1,\ldots,j_m)\in J$ and let $\bm i=(i_1,\dots,i_m)\in\Delta(\bm\nu)$. Write uniquely
\[
i_\ell=a_\ell+n_\ell q_\ell, \qquad 0\leq a_\ell<n_\ell,\qquad 0\leq q_\ell<t_\ell,
\]
for $\ell=1,\ldots,m$, and set $\bm a=(a_1,\ldots,a_m)$. Then, for every $z\in\mathbb L$,
\[
\N^{\bm i}(\bm\alpha^{(\bm j)}) \bm\sigma^{\bm i}(z)=\left(\prod_{\ell=1}^m\lambda_{\ell,j_\ell}^{q_\ell}\right)\N^{\bm a}(\bm\alpha^{(\bm j)}) \bm\sigma^{\bm a}(z).
\]
\end{lemma}
\begin{proof}
For every $\ell$, the automorphism $\sigma_\ell$ has order $n_\ell$ and
\[\prod_{h=0}^{n_\ell-1}\sigma_\ell^h(\alpha_{\ell,j_\ell})=\N_{\mathbb E_\ell/\mathbb K}(\alpha_{\ell,j_\ell})=\lambda_{\ell,j_\ell}.
\]
Hence, 
\[
\prod_{h=0}^{i_\ell-1}\sigma_\ell^h(\alpha_{\ell,j_\ell})=\lambda_{\ell,j_\ell}^{q_\ell}\prod_{h=0}^{a_\ell-1}\sigma_\ell^h(\alpha_{\ell,j_\ell}).
\]
Moreover, $\sigma_\ell^{i_\ell}=\sigma_\ell^{a_\ell},$ because $\sigma_\ell^{n_\ell}=\id_\mathbb L$. Multiplying these identities over
$\ell=1,\ldots,m$ gives the claim.
\end{proof}
\begin{theorem}\label{Moore=srk}
Let $v \in \mathbb{L}^{Nt}$. Then
\[
\srk(v) = \rk(M_{G,\bm{\alpha}}(v)).
\]
\end{theorem}
\begin{proof}
Let $s=\srk(v)$. We prove that $\rk_{\mathbb L}M_{G,\bm\alpha}(v)=s$.

Writing $v = \left(v^{(\bm{j})}\right)_{\bm{j} \in J}$, we have $s = \sum_{\bm{j} \in J} r_{\bm{j}}$ with $r_{\bm{j}} = \rk(v^{(\bm{j})}).$
Hence, for every $\bm{j} \in J$, there exists $A^{(\bm{j})} \in \GL_N(\mathbb{K})$  such that, if $A = \diag\left(A^{(\bm{j})}\right)_{\bm{j} \in J}$, then  
\[
    v \cdot A = \left( v^{(\bm{j})} \cdot A^{(\bm{j})} \right)_{\bm{j} \in J} = \left( u^{(\bm{j})} \right)_{\bm{j} \in J} = u,
\]
with $u^{(\bm{j})} = (u_{\bm{j},1}, \dots, u_{\bm{j},r_{\bm{j}}}, 0, \dots, 0)$ and $\rk(u^{(\bm{j})}) = r_{\bm{j}}$. 
Then, since $A$ is invertible, we obtain
\[
    \rk_{\mathbb L}M_{G,\bm\alpha}(v)=\rk_{\mathbb L}\bigl(M_{G,\bm\alpha}(v)\cdot A\bigr)
\]
and since $A$ has entries in $\mathbb{K}$, it is fixed by any $\sigma\in\Gal(\mathbb{L}/\mathbb{K})$ and
\[
    M_{G,\bm{\alpha}}(v) \cdot A = M_{G,\bm{\alpha}}(v \cdot A)= M_{G,\bm{\alpha}}(u).
\]
Therefore,
\[
    \rk_{\mathbb L}M_{G,\bm\alpha}(u)\le s
\]
because $u$ has $\sum_{\bm{j} \in J} r_{\bm{j}} = s$ nonzero columns. It remains to prove that these $s$ columns are linearly independent.
Let
\[
C:=\{(\bm j,k)\mid \bm j\in J,\ 1\leq k\leq r_{\bm{j}}\}
\]
be their index set. Suppose, by contradiction, that they are linearly
dependent, and choose a subset $S\subseteq C$ of minimum cardinality
among the dependent subsets. Thus, there are coefficients
$c_{\bm j,k}\in \mathbb L^*$, $(\bm j,k)\in S$, such that
\begin{equation}\label{eq:minimal-moore-dependence}
\sum_{(\bm j,k)\in S}c_{\bm j,k}\N^{\bm i}(\bm\alpha^{(\bm j)})\bm\sigma^{\bm i}(u_{\bm j,k})=0
    \qquad\text{for every }\bm i\in\Delta(\bm\nu).
\end{equation}
By minimality, every nonempty proper subset of $S$ indexes a linearly
independent family of full columns.
Fix $\ell\in\{1,\ldots,m\}$. For $h\in\{1,\ldots,t_\ell\}$ set $S_h:=\{(\bm j,k)\in S\mid j_\ell=h\}.$ Observe that these sets form a disjoint partition of $S$, namely
\[
S=\bigsqcup_{h=1}^{t_\ell}S_h.
\]
Now fix $\bm a=(a_1,\ldots,a_m)\in\Delta(\bm n)$ and fix $q_r\in\Delta(t_r)$ for every $r\neq\ell$. Let $q_\ell$ vary in $\Delta(t_\ell)$ and consider the row
\[
i_r=a_r+n_rq_r\qquad(r=1,\ldots,m).
\]
We may therefore separate the factor corresponding to the $\ell$-th coordinate and, using Lemma~\ref{lem:moore-periodicity}, we rewrite the sum in \eqref{eq:minimal-moore-dependence} as follows:
\[
\begin{aligned}
0
&=\sum_{h=1}^{t_\ell}\sum_{(\bm j,k)\in S_h}c_{\bm j,k}\lambda_{\ell,j_\ell}^{q_\ell}\left(\prod_{\substack{r=1\\r\neq\ell}}^{m}\lambda_{r,j_r}^{q_r}\right)\N^{\bm a}(\bm\alpha^{(\bm j)})\bm\sigma^{\bm a}(u_{\bm j,k})
\\
&=
\sum_{h=1}^{t_\ell}\lambda_{\ell,h}^{q_\ell}\left[\sum_{(\bm j,k)\in S_h}c_{\bm j,k}\left(\prod_{\substack{r=1\\ r\neq\ell}}^{m}\lambda_{r,j_r}^{q_r}\right)\N^{\bm a}(\bm\alpha^{(\bm j)})\bm\sigma^{\bm a}(u_{\bm j,k})
\right].
\end{aligned}
\]
This is a homogeneous linear system of $t_\ell$ equations
\[
\sum_{h=1}^{t_\ell}\lambda_{\ell,h}^{q_\ell} A_h=0,\qquad q_\ell=0,\ldots,t_\ell-1,
\]
where
\[
A_h:=\sum_{(\bm j,k)\in S_h}c_{\bm j,k}\left(\prod_{\substack{r=1\\ r\neq\ell}}^m\lambda_{r,j_r}^{q_r}\right)\N^{\bm a}(\bm\alpha^{(\bm j)})\bm\sigma^{\bm a}(u_{\bm j,k}).
\]
Its coefficient matrix is the Vandermonde matrix
\[
V_\ell:=\left(\lambda_{\ell,h}^{q_\ell}\right)_{\substack{0\leq q_\ell<t_\ell\\
1\leq h\leq t_\ell
}}.
\]
Since the elements $\lambda_{\ell,1},\ldots,\lambda_{\ell,t_\ell}$
are pairwise distinct, the matrix above is invertible over $\mathbb K$. It follows that
\[
A_h=0\qquad\text{for every }h.
\]
Since $\bm a$ and the $q_r$, $r\neq\ell$, were arbitrary, another
application of Lemma~\ref{lem:moore-periodicity} shows that, for every
$h$,
\[
\sum_{(\bm j,k)\in S_h}c_{\bm j,k}\N^{\bm i}(\bm\alpha^{(\bm j)})\bm\sigma^{\bm i}(u_{\bm j,k})=0 \qquad\text{for every }\bm i\in\Delta(\bm\nu).
\]
Thus, every nonempty $S_h$ itself indexes a dependence among full
columns. By the minimality of $S$, exactly one $S_h$ is nonempty.
Therefore, all elements $(\bm j,k)\in S$ have the same $\ell$-th block
coordinate $j_\ell$.
Repeating the argument for $\ell=1,\ldots,m$, we conclude that all
columns indexed by $S$ belong to one and the same block
$\bm j\in J$. This contradicts
Lemma \ref{lem:block-moore-independence}, because the nonzero columns
of each fixed block are linearly independent over $\mathbb L$.
Hence, the $s$ nonzero columns of $M_{G,\bm\alpha}(u)$ are linearly
independent, and therefore
\[
\rk_\mathbb L M_{G,\bm\alpha}(v)=\rk_\mathbb L M_{G,\bm\alpha}(u)=s=\srk(v).
\]
\end{proof}

\subsection{The Dickson--Moore rank correspondence}
We prove that the Dickson matrix of $F$ and the generalized Moore matrix of the vector obtained by evaluating $\bm\Phi(F)$ on a $\mathbb K$-basis of $\mathbb L$ have the same rank over $\mathbb L$.

Let $\mathcal B_{\ev}=(\beta_1,\ldots,\beta_N)$ be an ordered $\mathbb K$-basis of $\mathbb L$. We define
\[
\ev_{\mathcal B_{\ev}^{(t)}}\colon(\mathbb L[G])^t\longrightarrow\mathbb L^{Nt},
\qquad
(P_{\bm j})_{\bm j\in J}\longmapsto\bigl(P_{\bm j}(\beta_1),\ldots,P_{\bm j}(\beta_N)\bigr)_{\bm j\in J},
\]
where $P(\beta)$ denotes the natural action of $P\in\mathbb L[G]$ on $\mathbb L$. This map is an $\mathbb L$-linear bijection, and $\srk\bigl(\ev_{\mathcal B_{\ev}^{(t)}}(P)\bigr)=\srk(P)$ for every $P\in(\mathbb L[G])^t$.
\begin{theorem}\label{rkMoore=rkDickson}
Let $F\in\Lx/H_\Lambda$ and set $F_{\bm\alpha}:=\bm\Phi(F)\in(\mathbb L[G])^t$.
Let $v = \ev_{\mathcal B_{\ev}^{(t)}}(F_{\bm{\alpha}})\in\mathbb{L}^{Nt}$. Then
\[
\rk_{\,\mathbb{L}}(D_{\mathcal B_{\mathrm{mon}}}(F)) = \rk_{\,\mathbb{L}}(M_{G,\bm{\alpha}}(v)).
\]
\end{theorem}
\begin{proof}
Every element $g\in\Lx/H_\Lambda$ can be uniquely written as
\[
g=\sum_{\bm i\in\Delta(\bm\nu)} d_{\bm i} \bm{x}^{\bm{i}},
\qquad d_{\bm i}\in \mathbb{L},
\]
where $\bm{i}=(i_1,\dots,i_m)$ and $\bm{x}^{\bm{i}}=x_1^{i_1}\cdots x_m^{i_m}.$

We want to characterize when $g$ belongs to the kernel of $\mu_F$. By definition,
\[
g\in \ker(\mu_F)
\iff
\mu_F(g)=0
\iff
g\cdot F=0 \quad \text{in }\;\mathbb{L}[x_1,\dots,x_m;\sigma_1,\dots,\sigma_m]/H_\Lambda.
\]
This is equivalent to $g\cdot F\in H_\Lambda.$

By Theorem \ref{representation theorem}, $H_\Lambda=\ker\bm{\Phi}$, hence
\[
g\in \ker(\mu_F)
\iff
g\cdot F\in \ker\bm{\Phi}.
\]
Therefore,
\[
g\in \ker(\mu_F)
\iff
\pi_{\bm j}\bigl(\phi(g\cdot F)\bigr)\in I_{\bm n}
\qquad \text{for all } \bm{j}\in J.
\]
By substituting the expansion of $g$, we obtain
\[
\sum_{\bm i\in\Delta(\bm\nu)} d_{\bm i}\bm{x}^{\bm{i}} \in \ker(\mu_F)
\iff
\pi_{\bm{j}}\!\left(
\phi\!\left(\left(\sum_{\bm{i}\in\Delta(\bm\nu)} d_{\bm{i}}\bm{x}^{\bm{i}}\right)F\right)
\right)\in I_{\bm n}
\qquad \text{for all } \bm{j}\in J.
\]
Since $\phi$ and $\pi_{\bm{j}}$ are linear, this is equivalent to
\[
\sum_{\bm{i}\in\Delta(\bm\nu)} d_{\bm{i}}\,\pi_{\bm{j}}\bigl(\phi(\bm{x}^{\bm{i}}F)\bigr)\in I_{\bm n}
\qquad \text{for all } \bm{j}\in J.
\]
Using that $\phi$ is a ring homomorphism, we may rewrite
\[
\pi_{\bm{j}}\bigl(\phi(\bm{x}^{\bm{i}}F)\bigr)
=
\pi_{\bm{j}}\bigl(\phi(\bm{x}^{\bm{i}})\phi(F)\bigr)
=
\pi_{\bm{j}}\bigl(\phi(\bm{x}^{\bm{i}})\bigr)\,\pi_{\bm{j}}\bigl(\phi(F)\bigr),
\]
so that
\[
\sum_{\bm{i}\in\Delta(\bm\nu)} d_{\bm{i}}\,\pi_{\bm{j}}\bigl(\phi(\bm{x}^{\bm{i}})\bigr)\,\pi_{\bm{j}}\bigl(\phi(F)\bigr)\in I_{\bm n}
\qquad \text{for all } \bm{j}\in J.
\]
Now let $\bm{j}=(j_1,\dots,j_m)\in J$. By definition of $\phi$, we have
\[
\pi_{\bm{j}}\bigl(\phi(\bm{x}^{\bm{i}})\bigr)
=
\pi_{\bm{j}}\bigl(\phi(x_1^{i_1}\cdots x_m^{i_m})\bigr)
=
({\alpha}_{1,j_1}x_1)^{i_1}\cdots({\alpha}_{m,j_m}x_m)^{i_m}.
\]
Using the multiplication rule in the skew polynomial ring, for each $\ell\in\{1,\dots,m\}$ we obtain
\[
({\alpha}_{\ell, j_\ell}x_\ell)^{i_\ell}
=
\left(\prod_{k=0}^{i_\ell-1}\sigma_\ell^k\!\left({\alpha}_{\ell, j_\ell}\right)\right)x_\ell^{i_\ell}.
\]
Therefore,
\[
\pi_{\bm{j}}\bigl(\phi(\bm{x}^{\bm i})\bigr)
=
\prod_{\ell=1}^m
\left(\prod_{k=0}^{i_\ell-1}\sigma_\ell^k\!\left({\alpha}_{\ell, j_\ell}\right)\right)x_\ell^{i_\ell}.
\]
If we denote the truncated norm of $\alpha_{\ell,j_\ell}$ by
\[
\N_{\mathbb{E}_\ell/\mathbb K}^{\,i_\ell}\!\left({\alpha}_{\ell, j_\ell}\right)
:=
\prod_{k=0}^{i_\ell-1}\sigma_\ell^k\!\left({\alpha}_{\ell, j_\ell}\right),
\]
then the previous identity becomes
\[
\pi_{\bm{j}}\bigl(\phi(\bm{x}^{\bm{i}})\bigr)
=
\prod_{\ell=1}^m
\N_{\mathbb{E}_\ell/\mathbb K}^{\,i_\ell}\!\left({\alpha}_{\ell, j_\ell}\right)x_\ell^{i_\ell}.
\]
Consequently,
\[
\sum_{\bm{i}\in\Delta(\bm\nu)}
d_{\bm{i}}\,
\N^{\bm{i}}(\bm{\alpha}^{(\bm{j})})\,\bm{x}^{\bm{i}}
\pi_{\bm{j}}(\phi(F))
\in I_{\bm n}
\qquad \text{for all } \bm{j}\in J,
\]
where
\[
\N^{\bm{i}}(\bm{\alpha}^{(\bm j)})\,\bm{x}^{\bm{i}}=\prod_{\ell=1}^m
\N_{\mathbb{E}_\ell/\mathbb K}^{\,i_\ell}\!\left({\alpha}_{\ell, j_\ell}\right)x_\ell^{i_\ell}.
\]
Since $c\,\bm x^{\bm i}=\bm x^{\bm i}\,\bm\sigma^{-\bm i}(c)$ for $c\in\mathbb L$, this condition reads
\[
\sum_{\bm i\in\Delta(\bm\nu)}d_{\bm i}\,\bm x^{\bm i}\,\bm\sigma^{-\bm i}\bigl(\N^{\bm i}(\bm\alpha^{(\bm j)})\bigr)\,\phi_{\bm j}(F)\in I_{\bm n}
\qquad\text{for all }\bm j\in J.
\]
By Theorem~\ref{Theorem 60 ACLN}, this means that $\psi$ maps the left-hand side to $0$. An element of $\mathbb L[G]$ is zero if and only if it vanishes on the basis $\mathcal B_{\ev}$. Put $v_{\bm j,k}:=\psi\bigl(\phi_{\bm j}(F)\bigr)(\beta_k)$, so that $v=(v_{\bm j,k})_{\bm j,k}=\ev_{\mathcal B_{\ev}^{(t)}}(F_{\bm\alpha})$. Since $\psi(\bm x^{\bm i}cP)(z)=\bm\sigma^{\bm i}\bigl(c\,\psi(P)(z)\bigr)$ for $c,z\in\mathbb L$ and $P\in\Lx$, we have
\[
\psi\Bigl(\bm x^{\bm i}\,\bm\sigma^{-\bm i}\bigl(\N^{\bm i}(\bm\alpha^{(\bm j)})\bigr)\,\phi_{\bm j}(F)\Bigr)(\beta_k)
=\N^{\bm i}(\bm\alpha^{(\bm j)})\,\bm\sigma^{\bm i}(v_{\bm j,k}).
\]
Therefore, $g=\sum_{\bm i\in\Delta(\bm\nu)}d_{\bm i}\bm x^{\bm i}$ lies in $\ker(\mu_F)$ if and only if
\[
\sum_{\bm i\in\Delta(\bm\nu)}d_{\bm i}\,\N^{\bm i}(\bm\alpha^{(\bm j)})\,\bm\sigma^{\bm i}(v_{\bm j,k})=0
\qquad\text{for all }(\bm j,k)\in J\times\{1,\ldots,N\},
\]
that is, if and only if the row vector $(d_{\bm i})_{\bm i\in\Delta(\bm\nu)}$ lies in the left kernel of $M_{G,\bm\alpha}(v)$. Hence, $\dim_{\mathbb L}\ker(\mu_F)=Nt-\rk_{\mathbb L}M_{G,\bm\alpha}(v)$. On the other hand, $\dim_{\mathbb L}\ker(\mu_F)=Nt-\rk_{\mathbb L}D_{\mathcal B_{\mathrm{mon}}}(F)$ by Remark~\ref{rk=dimrispcolonideal}, which proves the claim.
\end{proof}

\begin{corollary}\label{corollario moore dickson e srk}
With the notation of Theorem~\ref{rkMoore=rkDickson},
\[
\srk\bigl(\bm\Phi(F)\bigr)=\srk(v)=\rk_{\mathbb L}D_{\mathcal B_{\mathrm{mon}}}(F)=\rk_{\mathbb L}M_{G,\bm\alpha}(v).
\]
\end{corollary}
\begin{proof}
The first equality holds because $\ev_{\mathcal B_{\ev}^{(t)}}$ preserves the sum-rank weight; the others follow from Theorems~\ref{Moore=srk} and~\ref{rkMoore=rkDickson}.
\end{proof}

\begin{example}[Continuation of Example~\ref{ex:running-example}]\label{ex:running-moore}
 Let $F=1+x_1+x_2+x_1x_2\in\mathcal{R}$ and consider
$F_{\bm{\alpha}}=\bm\Phi(F)\in(\mathbb L[G])^4.
$

For every $\bm{j}=(j_1,j_2)\in\{1,2\}^2$, the corresponding
component of $F_{\bm{\alpha}}$ is
\[
F_{\bm{\alpha}}^{(\bm{j})}
=
\bigl(1+\alpha_{1,j_1}\sigma_1\bigr)
\bigl(1+\alpha_{2,j_2}\sigma_2\bigr).
\]
Equivalently,
\[
F_{\bm{\alpha}}^{(\bm{j})}
=
1+\alpha_{1,j_1}\sigma_1
+\alpha_{2,j_2}\sigma_2
+\alpha_{1,j_1}\alpha_{2,j_2}
 \sigma_1\sigma_2.
\]
We choose the $\mathbb K$-basis of $\mathbb L/\mathbb K$
\[
\mathcal B_{\ev}=(\beta_1,\beta_2,\beta_3,\beta_4)=(1,\sqrt2,\sqrt3,\sqrt6)
\]
and set $\mathcal B_{\ev}^{(4)}=(\mathcal B_{\ev})^4.$
Let
\[
v
=
\ev_{\mathcal B_{\ev}^{(4)}}
(F_{\bm{\alpha}})
=
\bigl(
v^{(1,1)},
v^{(1,2)},
v^{(2,1)},
v^{(2,2)}
\bigr)\in(\mathbb {L}^4)^4.
\]
By evaluating the four components of
$F_{\bm{\alpha}}$ on
$\mathcal B_{\ev}$, we obtain
\[
v^{(1,1)}=(4,0,0,0),\quad v^{(1,2)}=(4+2\sqrt3,\;0,\;-6,\;0),
\]
\[
v^{(2,1)}=(4+2\sqrt2,\;-4,\;0,\;0),
\quad
v^{(2,2)}=\bigl(4+2\sqrt2+2\sqrt3+\sqrt6,-4-2\sqrt3,-6-3\sqrt2,6\bigr).
\]
The associated Moore matrix has the block decomposition
\[
M_{G,\bm{\alpha}}(v)
=
\left(
M^{(1,1)}
\;\middle|\;
M^{(1,2)}
\;\middle|\;
M^{(2,1)}
\;\middle|\;
M^{(2,2)}
\right)
\in\mathbb L^{16\times16},
\]
where each block $M^{(j_1,j_2)}$ belongs to
$\mathbb L^{16\times4}$ and is defined by
\[
M^{(j_1,j_2)}_{(i_1,i_2),k}
=
\N^{(i_1,i_2)}
\bigl(
\alpha_{1,j_1},\alpha_{2,j_2}
\bigr)
\sigma_1^{i_1}\sigma_2^{i_2}
\bigl(v^{(j_1,j_2)}_k\bigr),
\]
for $(i_1,i_2)\in\Delta(4)\times\Delta(4),$ and $k\in\{1,2,3,4\}.$
We now compute the truncated norms appearing in the definition of the Moore matrix. Recall that, for
\[ 
\bm{i}=(i_1,i_2)\in\{0,1,2,3\}^2,\qquad \bm{\alpha}^{(\bm{j})} = \left( \alpha_{1,j_1},\alpha_{2,j_2} \right), \qquad \bm{j}=(j_1,j_2)\in\{1,2\}^2, 
\] 
we have 
\[ 
\N^{(i_1,i_2)} \left( \alpha_{1,j_1},\alpha_{2,j_2} \right) = \left( \prod_{h=0}^{i_1-1} \sigma_1^h\left(\alpha_{1,j_1}\right) \right) \left( \prod_{\ell=0}^{i_2-1} \sigma_2^\ell\left(\alpha_{2,j_2}\right) \right).
\] 
As usual, an empty product is equal to $1$.
Let us first consider the elements associated with the first variable. Since $\alpha_{1,1}=1 $ and every field automorphism fixes $1$, we obtain 
\[ 
\left( \prod_{h=0}^{i_1-1} \sigma_1^h\left(\alpha_{1,1}\right) \right)_{i_1=0}^3 = (1,1,1,1).
\] 
For $ \alpha_{1,2}=1+\sqrt{2}, $ we have 
\[ 
\sigma_1(1+\sqrt{2})=1-\sqrt{2}\qquad\text{and}\qquad \sigma_1^2=\id_{\mathbb L}.
\] 
Therefore, 
\[ 
\prod_{h=0}^{-1} \sigma_1^h(1+\sqrt{2}) = 1, 
\] 
\[ 
\prod_{h=0}^{0} \sigma_1^h(1+\sqrt{2}) = 1+\sqrt{2}, 
\] 
\[
\prod_{h=0}^{1} \sigma_1^h(1+\sqrt{2})= (1+\sqrt{2})(1-\sqrt{2})=-1,
\] 
\[
\prod_{h=0}^{2} \sigma_1^h(1+\sqrt{2})= (1+\sqrt{2})(1-\sqrt{2}) \sigma_1^2(1+\sqrt{2})= (-1)(1+\sqrt{2})=-1-\sqrt{2}.
\] 
Hence, 
\[
\left( \prod_{h=0}^{i_1-1} \sigma_1^h\left(\alpha_{1,2}\right) \right)_{i_1=0}^3 = (1,\,1+\sqrt{2},\,-1,\,-1-\sqrt{2}).
\] 
We proceed analogously for the second variable. 

The truncated norm $ \N^{\bm{i}} \left( \bm{\alpha}^{(\bm{j})} \right) $ is obtained by multiplying the value associated with the first coordinate by the corresponding one associated with the second coordinate. For instance, let 
\[ \bm{j}=(2,2), \qquad \bm{\alpha}^{(2,2)} = (1+\sqrt{2},1+\sqrt{3}),\qquad\text{and}\qquad \bm{i}=(2,3).
\] 
Then 
\[ 
\N^{(2,3)} \left( \bm{\alpha}^{(2,2)} \right) = \left( \prod_{h=0}^{1} \sigma_1^h(1+\sqrt{2}) \right) \left( \prod_{\ell=0}^{2} \sigma_2^\ell(1+\sqrt{3}) \right)= (-1)(-2-2\sqrt{3})= 2+2\sqrt{3}.
\] 
The general entry of the Moore matrix is 
\[ 
M_{G,\bm{\alpha}}(v)_{ \bm{i},(\bm{j},k)} = \N^{\bm{i}} \left( \bm{\alpha}^{(\bm{j})} \right) \sigma^{\bm{i}} \left(v_{\bm{j},k}\right),
\qquad \text{where}\quad 
\sigma^{\bm{i}} = \sigma_1^{i_1}\sigma_2^{i_2}.
\] 
Consequently, for the particular choice $\bm{i}=(2,3)$ and $\bm{j}=(2,2),$ we have $\sigma^{(2,3)} = \sigma_1^2\sigma_2^3 = \sigma_2,$ and hence 
\[ 
M_{G,\bm{\alpha}}(v)_{ (2,3),((2,2),k)} = (2+2\sqrt{3})\, \sigma_2\left(v_{(2,2),k}\right)\;,\qquad k\in\{1,2,3,4\}.
\]
Since $v^{(1,1)}=(4,0,0,0)$ and
$\alpha_{1,1}=\alpha_{2,1}=1$, its first block is
\[
M^{(1,1)}
=
4
\begin{pmatrix}
1&0&0&0\\
1&0&0&0\\
\vdots&\vdots&\vdots&\vdots\\
1&0&0&0
\end{pmatrix}
\in\mathbb L^{16\times4}
\qquad
\text{and}\quad \rk_{\mathbb L}\bigl(M^{(1,1)}\bigr)=1.
\]
For readability, set $a:=\sqrt{2}$ and $b:=\sqrt{3}$. The vertical lines in the matrix below separate the four blocks corresponding, respectively, to \[ (1,1),\qquad (1,2),\qquad (2,1),\qquad (2,2). \] \begingroup \setlength{\arraycolsep}{2.5pt} \renewcommand{\arraystretch}{1.15} \begin{center} \resizebox{\textwidth}{!}{$ M_{G,\bm{\alpha}}(v) = \left( \begin{array}{cccc|cccc|cccc|cccc} 4&0&0&0 &4+2b&0&-6&0 &4+2a&-4&0&0 &4+2a+2b+ab&-4-2b&-6-3a&6 \\ 4&0&0&0 &-2+2b&0&-6-6b&0 &4+2a&-4&0&0 &-2-a+2b+ab&2-2b&-6-3a-6b-3ab&6+6b \\ 4&0&0&0 &-8-4b&0&12&0 &4+2a&-4&0&0 &-8-4a-4b-2ab&8+4b&12+6a&-12 \\ 4&0&0&0 &4-4b&0&12+12b&0 &4+2a&-4&0&0 &4+2a-4b-2ab&-4+4b&12+6a+12b+6ab&-12-12b \\[1mm] 4&0&0&0 &4+2b&0&-6&0 &2a&-4-4a&0&0 &2a+ab&-4-4a-2b-2ab&-3a&6+6a \\ 4&0&0&0 &-2+2b&0&-6-6b&0 &2a&-4-4a&0&0 &-a+ab&2+2a-2b-2ab&-3a-3ab&6+6a+6b+6ab \\ 4&0&0&0 &-8-4b&0&12&0 &2a&-4-4a&0&0 &-4a-2ab&8+8a+4b+4ab&6a&-12-12a \\ 4&0&0&0 &4-4b&0&12+12b&0 &2a&-4-4a&0&0 &2a-2ab&-4-4a+4b+4ab&6a+6ab&-12-12a-12b-12ab \\[1mm] 4&0&0&0 &4+2b&0&-6&0 &-4-2a&4&0&0 &-4-2a-2b-ab&4+2b&6+3a&-6 \\ 4&0&0&0 &-2+2b&0&-6-6b&0 &-4-2a&4&0&0 &2+a-2b-ab&-2+2b&6+3a+6b+3ab&-6-6b \\ 4&0&0&0 &-8-4b&0&12&0 &-4-2a&4&0&0 &8+4a+4b+2ab&-8-4b&-12-6a&12 \\ 4&0&0&0 &4-4b&0&12+12b&0 &-4-2a&4&0&0 &-4-2a+4b+2ab&4-4b&-12-6a-12b-6ab&12+12b \\[1mm] 4&0&0&0 &4+2b&0&-6&0 &-2a&4+4a&0&0 &-2a-ab&4+4a+2b+2ab&3a&-6-6a \\ 4&0&0&0 &-2+2b&0&-6-6b&0 &-2a&4+4a&0&0 &a-ab&-2-2a+2b+2ab&3a+3ab&-6-6a-6b-6ab \\ 4&0&0&0 &-8-4b&0&12&0 &-2a&4+4a&0&0 &4a+2ab&-8-8a-4b-4ab&-6a&12+12a \\ 4&0&0&0 &4-4b&0&12+12b&0 &-2a&4+4a&0&0 &-2a+2ab&4+4a-4b-4ab&-6a-6ab&12+12a+12b+12ab \end{array} \right). $} \end{center} \endgroup 
Notice that the first row coincides with $v$. Indeed, for $\bm{i}=(0,0)$, one has $\N^{(0,0)} \left(\bm{\alpha}^{(\bm{j})}\right)=1$ and $\sigma^{(0,0)}=\id_{\mathbb L}$, so $M_{G,\bm{\alpha}}(v)_{ (0,0),(\bm{j},k)} = v_{\bm{j},k}.$ The $\mathbb K$-ranks of the four blocks of $v$ are
\[
\rk_{\mathbb K}(v^{(1,1)})=1,
\quad
\rk_{\mathbb K}(v^{(1,2)})=2,
\quad
\rk_{\mathbb K}(v^{(2,1)})=2,
\quad
\rk_{\mathbb K}(v^{(2,2)})=4.
\]
Consequently, $\srk(v)=1+2+2+4=9.$ By Corollary \ref{corollario moore dickson e srk},
\[
\rk_{\mathbb L}
\bigl(M_{G,\bm{\alpha}}(v)\bigr)
=
9
=
\rk_{\mathbb L}
\bigl(D_{\mathcal B_{\mathrm{mon}}}(F)\bigr).
\]

\end{example}

\section{Sum-rank Reed--Muller-type codes}\label{sezione codici}
In Section~\ref{sec: rapp theorem}, more precisely in Theorem \ref{representation theorem}, we proved that the sum-rank ambient space $(\mathbb L[G])^t$
can be represented as a quotient of the multivariate skew polynomial ring
$\mathbb L[x_1,\ldots,x_m;\sigma_1,\ldots,\sigma_m]$. The following remark will clarify the suitable ambient space for our Reed--Muller-type construction.
\begin{remark}\label{rem:metric-quotient}
For $F_1,F_2\in\Lx/H_\Lambda$ put
\[
d_\Lambda(F_1,F_2):=\dim_{\mathbb L}\frac{\Lx}{(H_\Lambda:F_2-F_1)},
\qquad
\srk(F):=d_\Lambda(F,0).
\]
These quantities are well defined, because $(H_\Lambda:F)$ depends only on the class of $F$. Let $\mathcal B_{\ev}$ be an ordered $\mathbb K$-basis of $\mathbb L$, and let $\Theta_{\mathcal B_{\ev}}:=\rho_{\mathcal B_{\ev}^{(t)}}\circ\bm\Phi\colon\Lx/H_\Lambda\to(\mathbb K^{N\times N})^t$, with $\rho_{\mathcal B_{\ev}^{(t)}}$ as in Section~\ref{sec:preliminaries}. By Remark~\ref{rk=dimrispcolonideal} and Corollary~\ref{corollario moore dickson e srk},
\[
\srk(F)=\rk_{\mathbb L}D_{\mathcal B_{\mathrm{mon}}}(F)=\srk\bigl(\bm\Phi(F)\bigr)=\srk\bigl(\Theta_{\mathcal B_{\ev}}(F)\bigr).
\]
Hence, $d_\Lambda$ is a metric, and $\Theta_{\mathcal B_{\ev}}$ is a $\mathbb K$-linear isometry from $(\Lx/H_\Lambda,d_\Lambda)$ onto $\bigl((\mathbb K^{N\times N})^t,d_{\srk}\bigr)$. All codes below are $\mathbb L$-subspaces of $\Lx/H_\Lambda$ with this metric.
\end{remark}
\begin{definition}\label{def: srk reed muller codes}
Let $r$ be a nonnegative integer. The \textbf{sum-rank Reed--Muller-type code} of order $r$ is defined as
\[
\SRRM_{\mathbb L/\mathbb K}(r, m)
:=
\left\{
F \bmod (H_1,\ldots,H_m)
\;\middle|\;
\deg F\le r
\right\}.
\]
Equivalently,
\[
\SRRM_{\mathbb L/\mathbb K}(r, m)=\Span_{\mathbb L}\left\{\bm x^{\bm i} \bmod (H_1,\ldots,H_m)\;\middle|\;0\le i_h<\nu_h,\ |\bm i|\le r\right\}.
\]
\end{definition}

\begin{remark}
For every integer $r\ge \sum_{i=1}^m(\nu_i-1)$, $\SRRM_{\mathbb L/\mathbb K}(r, m)$ coincides with the full ambient space and, hence, has minimum sum-rank distance $1$. Unless otherwise stated, we therefore assume that the order $r$
satisfies $0\le r<\sum_{i=1}^m(\nu_i-1)$.
\end{remark}

\begin{proposition}
The code $\SRRM_{\mathbb L/\mathbb K}(r, m)$ is an $\mathbb L$-linear code in the sum-rank ambient space
$(\mathbb L[G])^t$. Its sum-rank length is
\[
Nt=\prod_{i=1}^m \nu_i.
\]
Moreover,
\[
\dim_{\mathbb L} \SRRM_{\mathbb L/\mathbb K}(r, m)
=
\#\left\{
\bm i=(i_1,\ldots,i_m)\in\mathbb N^m
\;\middle|\;
0\le i_h<\nu_h,\ |\bm i|\le r
\right\}.
\]
\end{proposition}
\begin{proof}
Reducing a monomial modulo $H_i$ replaces $x_i^{\nu_i}$ by a $\mathbb K$-linear combination of powers $x_i^a$ with $a<\nu_i$, so it does not increase the total degree. Hence, the two descriptions in Definition~\ref{def: srk reed muller codes} agree, and the claims follow from Theorem~\ref{representation theorem} and Remark~\ref{rem:universal-grobner-basis}.
\end{proof}
The dimension can also be written as follows.
\begin{proposition}\label{prop: dim srk muller codes}
We have
\[
\dim_{\mathbb L} \SRRM_{\mathbb L/\mathbb K}(r, m)
=
k(r,\bm{\nu}),
\]
where
\[
k(r,\bm{\nu}):=
\sum_{\ell=0}^r c(\ell,\bm{\nu})
=
\sum_{\ell=0}^r
[z^\ell]
\prod_{j=1}^m
\left(\frac{1-z^{\nu_j}}{1-z}\right),
\]
and $c(\ell,\bm{\nu})$ is the number of weak compositions of
the integer $\ell$ into $m$ parts in which the $j$-th part is at most
$\nu_j-1$, and $[z^\ell]p(z)$ denotes the coefficient of $z^\ell$ in
$p(z)$.
\end{proposition}

\begin{example}[Continuation of Example~\ref{ex:running-example}]
For $r=2$, the associated sum-rank Reed--Muller-type code is
\[
\SRRM_{\mathbb L/\mathbb K}(2,2)=\Span_{\mathbb L}\{1,x_1,x_2,x_1^2,x_1x_2,x_2^2\}.
\]
In particular, $\dim_{\mathbb L}\SRRM_{\mathbb L/\mathbb K}(2,2)=6.$
\end{example}
\subsection{Degree bounds for the sum-rank weight}
In this section we derive degree bounds for the sum-rank weight of reduced
skew polynomial representatives and we compute the minimum distance of the
codes introduced in the previous section.
\begin{notation}\label{notation min dist}
Up to reordering the variables, we will assume from now on that $\nu_1\ge \nu_2\ge \cdots \ge \nu_m.$
For an integer $r$ with $0\le r<\sum_{i=1}^m(\nu_i-1),$ we define
\[
d(r,\bm{\nu})
:=
\min
\left\{
\prod_{i=1}^m(\nu_i-u_i)
\;\middle|\;
\bm{u}=(u_1,\ldots,u_m)\in\Delta(\bm{\nu}),\ |\bm{u}|\le r
\right\}.
\]
\end{notation}
We first recall the closed form of $d(r,\bm\nu)$. Since $\prod_i(\nu_i-u_i)$ decreases when a coordinate of $\bm u$ increases, the minimum defining $d(r,\bm\nu)$ is attained at some $\bm u$ with $|\bm u|=r$. With $y_i:=\nu_i-u_i$, it is the minimum of $y_1\cdots y_m$ subject to $1\le y_i\le\nu_i$ and $\sum_iy_i=\sum_i\nu_i-r$, which is computed in the following lemma.
\begin{lemma}[{\cite[Lemma 2.2]{CFS2017}, \cite[Lemma 2.1]{BCPS2018}}]\label{lem:distance-function}
Let $0\le r< \sum_{i=1}^m(\nu_i-1).$ Let $s$ and $\ell$ be uniquely determined by
\[
r=\sum_{i=s+1}^m(\nu_i-1)+\ell,
\qquad
0\le \ell< \nu_s-1.
\]
Then
\[
d(r,\bm{\nu})
=
(\nu_s-\ell)\prod_{i=1}^{s-1}\nu_i.
\]
\end{lemma}

\begin{theorem}[Alon--Füredi-type bound]\label{thm:sum-rank-alon-furedi}
Let $0\neq F$ be a reduced representative in
\[
\frac{\mathbb L[x_1,\ldots,x_m;\sigma_1,\ldots,\sigma_m]}
{(H_1(x_1),\ldots,H_m(x_m))}
\]
with $0\le \deg(F)< \sum_{i=1}^m(\nu_i-1)$. Then
\[
\srk(F)
\ge
(\nu_s-\ell)\prod_{i=1}^{s-1}\nu_i,
\]
where $\ell$ and $s$ are the unique integers satisfying
\[\deg F=\sum_{i=s+1}^m(\nu_i-1)+\ell,
\qquad
0\le \ell<\nu_s-1,\quad 1\le s\le m.
\]
\end{theorem}
\begin{proof}
This proof follows the same approach used in \cite[Theorem 41]{ACLN}. Fix a monomial order refining the total degree and let $\LM(F)=\bm x^{\bm u}$ be the leading monomial of $F$, with $\bm u=(u_1,\ldots,u_m)\in\Delta(\bm{\nu}).$ Set $r:=\deg(F)$. Since the order refines the total degree, $|\bm u|=r$.
Consider the $\mathbb L$-subspace
\[
\mathcal T_{\bm u}:=\Span_{\mathbb L}\left\{\bm x^{\bm v}\;\middle|\;0\le v_i<\nu_i-u_i,\ i=1,\ldots,m\right\}
\]
inside the quotient $\mathbb L[x_1,\ldots,x_m;\sigma_1,\ldots,\sigma_m]/(H_1,\ldots,H_m)$.
We claim that right multiplication by $F$ is injective on
$\mathcal T_{\bm u}$. Let $0\neq g\in\mathcal T_{\bm u}$, and write
\[
\LT(g)=b_{\bm v}\bm x^{\bm v},
\qquad
\LT(F)=a_{\bm u}\bm x^{\bm u},
\]
with $a_{\bm u},b_{\bm v}\in \mathbb L^*$. Since multiplication in the skew polynomial
ring only twists the coefficients, we have
\[
\LT(gF)
=
b_{\bm v}\bm\sigma^{\bm v}(a_{\bm u})\bm x^{\bm v+\bm u}.
\]
In particular, its coefficient is nonzero.
Since $g\in\mathcal T_{\bm u}$, we have $v_i+u_i<\nu_i$ for every $i\in\{1,\ldots,m\}.$
Hence, $\bm x^{\bm v+\bm u}$ is already reduced modulo
$(H_1,\ldots,H_m)$.

Moreover, every other monomial occurring in $gF$ is strictly smaller
than $\bm x^{\bm v+\bm u}$. If such a monomial has to be reduced modulo some
$H_i$, its total degree strictly decreases, since $H_i$ is monic of
degree $\nu_i$ in $x_i$. As the chosen monomial order refines the
total degree, no reduction of the remaining terms can produce the
monomial $\bm x^{\bm v+\bm u}$. Therefore, the coefficient
$b_{\bm v}\bm\sigma^{\bm v}(a_{\bm u})$ of $\bm x^{\bm v+\bm u}$ remains nonzero in the reduced
representative of $gF$.
It follows that $gF\neq 0$ in the quotient. Thus, right multiplication by $F$ is injective on $\mathcal T_{\bm u}$.
It follows that the rank of the Dickson matrix associated with the right multiplication map by $F$ is at least
\[
\dim_{\mathbb L}\mathcal T_{\bm u}
=
\prod_{i=1}^m(\nu_i-u_i).
\]
By the rank identity proved in the previous section, $\srk(F)
=\rk_{\mathbb L}(D_{\mathcal B_{\mathrm{mon}}}(F)).$
Thus,
\[
\srk(F)
\ge
\prod_{i=1}^m(\nu_i-u_i).
\]
The statement of the theorem is now an immediate consequence of Lemma \ref{lem:distance-function} by taking the minimum over all $\bm u\in\Delta(\bm \nu)$ with $|\bm u|\le r$.
\end{proof}
\begin{corollary}[Schwartz--Zippel-type bound]\label{cor:sum-rank-schwartz-zippel}
Let $F\neq0$ be a reduced representative, and let $\bm x^{\bm u}=\LM(F)$ with respect to a monomial order refining the total degree. Then
\[
\frac{\srk(F)}{Nt}\ \ge\ \prod_{i=1}^m\Bigl(1-\frac{u_i}{\nu_i}\Bigr)\ \ge\ 1-\sum_{i=1}^m\frac{u_i}{\nu_i}\ \ge\ 1-\frac{\deg F}{\nu_m}.
\]
\end{corollary}
\begin{proof}
The proof of Theorem~\ref{thm:sum-rank-alon-furedi} shows that $\srk(F)\ge\prod_{i=1}^m(\nu_i-u_i)$; dividing by $Nt=\prod_{i=1}^m\nu_i$ gives the first inequality. The second holds because $0\le u_i/\nu_i\le1$, and the third because $\nu_m\le\nu_i$ for every $i$ and $|\bm u|=\deg F$.
\end{proof}
For $\deg F<\nu_m$ the last bound reads $\srk(F)\ge(\nu_m-\deg F)\prod_{i=1}^{m-1}\nu_i$, which coincides with the bound of Theorem~\ref{thm:sum-rank-alon-furedi} in this range.
\begin{lemma}\label{lem:existence monic right divisor}
For every $i\in\{1,\ldots,m\}$ and every integer $d$ with $0\le d\le \nu_i$,  the polynomial
\[
H_i(x_i) = \prod_{j=1}^{t_i}(x_i^{n_i}-\lambda_{i,j})
\]
has a monic right divisor $P_i(x_i)\in \mathbb E_i[x_i;\sigma_i]$
of degree $d$.
\end{lemma}
\begin{proof}
Fix $i\in\{1,\ldots,m\}$ and set $h_{i,j}(x_i):=x_i^{n_i}-\lambda_{i,j}$ for $j\in\{1,\ldots,t_i\}.$

We first show that, for every $j$ and every integer $r$ with $0\le r\le n_i$, the polynomial $h_{i,j}$ has a monic right divisor of degree $r$ in $\mathbb E_i[x_i;\sigma_i]$. Let $V\subseteq\mathbb E_i$ be a $\mathbb K$-subspace of dimension $r$, and let $\mathcal A_V$ be the left ideal of all $A=\sum_ka_kx_i^k\in\mathbb E_i[x_i;\sigma_i]$ whose operator evaluation $A(z):=\sum_ka_k\sigma_i^k(z)$ vanishes on $V$. Since $\mathbb E_i[x_i;\sigma_i]$ has a right division algorithm, $\mathcal A_V$ is generated by a unique monic $A_V$, and $x_i^{n_i}-1\in\mathcal A_V$. By Corollary~\ref{Corollary Thm 60 ACLN} for the cyclic extension $\mathbb E_i/\mathbb K$, the quotient $\mathcal A_V/(x_i^{n_i}-1)$ corresponds to the endomorphisms of $\mathbb E_i$ that vanish on $V$; these form a left $\mathbb E_i$-subspace of $\End_{\mathbb K}(\mathbb E_i)$ of dimension $n_i-r$. Hence
\[
\deg A_V=\dim_{\mathbb E_i}\frac{\mathbb E_i[x_i;\sigma_i]}{\mathcal A_V}=n_i-(n_i-r)=r,
\]
and there exists $B_V(x_i)\in \mathbb E_i[x_i;\sigma_i]$ such that
\begin{equation}\label{annihilator Ei right division}
x_i^{n_i}-1=B_V(x_i)A_V(x_i).
\end{equation}
Consider the automorphism
\[
\tau_{i,j}\colon \mathbb E_i[x_i;\sigma_i]
\longrightarrow \mathbb E_i[x_i;\sigma_i]
\qquad \text{where}\quad \tau_{i,j}(f(x_i))=f(\alpha_{i,j}^{-1}x_i).
\]
Observe that $\tau_{i,j}(x_i^{n_i}-1)=\lambda_{i,j}^{-1}\bigl(x_i^{n_i}-\lambda_{i,j}\bigr)$. Thus, by applying $\tau_{i,j}$ to \eqref{annihilator Ei right division} we obtain
\[
\lambda_{i,j}^{-1}
\bigl(x_i^{n_i}-\lambda_{i,j}\bigr)
=
\tau_{i,j}(B_V)\tau_{i,j}(A_V).
\]
Let $c\in\mathbb E_i^*$ be the leading coefficient of $\tau_{i,j}(A_V)$, and put $R_V:=c^{-1}\tau_{i,j}(A_V)$, which is monic of degree $r$. Then $h_{i,j}=\bigl(\lambda_{i,j}\tau_{i,j}(B_V)\,c\bigr)R_V$, so $R_V$ is a monic right divisor of $h_{i,j}$ of degree $r$.

We now prove the statement for $H_i$. If $d=\nu_i$, we simply take $P_i=H_i$. Hence, assume that $d<\nu_i$. Write
\[
d=qn_i+r,
\qquad
0\le q<t_i,
\qquad
0\le r<n_i.
\]
By the preceding argument, we can choose a monic right divisor $R(x_i)$ of $h_{i,q+1}(x_i)$
of degree $r$. Thus, there exists a polynomial $S(x_i)$ such that
\[
h_{i,q+1}(x_i)=S(x_i)R(x_i).
\]
Define
\[
C(x_i)
=
\prod_{j=1}^{q}h_{i,j}(x_i)
\quad\text{and}\quad
P_i(x_i)=C(x_i)R(x_i).
\]
Since $C$ and $R$ are monic, $\deg P_i=qn_i+r=d$ and $P_i$ is monic.

It remains to show that $P_i$ right-divides $H_i$. Set
\[
D(x_i)=\prod_{j=q+2}^{t_i}h_{i,j}(x_i).
\]
The polynomials $h_{i,j}$ are central, and therefore $C$ and $D$ commute with every element of $\mathbb E_i[x_i;\sigma_i]$. Consequently,
\[
H_i=C\,h_{i,q+1}\,D=D\,h_{i,q+1}\,C=
D\,S\,R\,C=D\,S\,C\,R=D\,S\,P_i.
\]
Finally, $P_i$ right-divides $H_i$, as desired.
\end{proof}
We now prove that the lower bound in Theorem \ref{thm:sum-rank-alon-furedi} is sharp.

\begin{theorem}\label{thm:minimum-distance}
Let $0\le r<\sum_{i=1}^{m}(\nu_i-1)$ and assume $\nu_1\ge\nu_2\ge\cdots\ge\nu_m$.

Then the minimum sum-rank distance of
$\SRRM_{\mathbb L/\mathbb K}(r,m)$ is
\[
d_{\min}
\bigl(
\SRRM_{\mathbb L/\mathbb K}(r,m)
\bigr)
=
d(r,\bm\nu).
\]
\end{theorem}
\begin{proof}
Let $s$ and $\ell$ be determined uniquely by 
\[
r
=
\sum_{i=s+1}^{m}(\nu_i-1)+\ell,
\qquad
0\le\ell<\nu_s-1.
\]
For each $i \in \{s+1,\ldots, m\}$, choose a monic right divisor $P_i(x_i)\in \mathbb E_i[x_i;\sigma_i]$ of $H_i(x_i)$ such that $\deg P_i=\nu_i-1$.
Similarly, select a  monic right divisor $P_s(x_s)\in \mathbb E_s[x_s;\sigma_s]$ of $H_s(x_s)$ such that $\deg P_s=\ell$.

For $i<s$, set $P_i=1$. For every $i$, let $Q_i(x_i)$ be the polynomial such that $H_i(x_i)=Q_i(x_i)P_i(x_i).$
Define
\[
F
=
\prod_{i=1}^{m}P_i(x_i)
=
\prod_{i=s}^{m}P_i(x_i).
\]
Since
\[
\deg F=\ell+\sum_{i=s+1}^{m}(\nu_i-1)=r,
\]
$F$ is reduced and represents a nonzero codeword of $\SRRM_{\mathbb L/\mathbb K}(r,m)$.
Consider the left ideal
\[
\mathcal J = \sum_{i=1}^{m}\Lx Q_i.
\]
We claim that $\mathcal J\subseteq(H_\Lambda:F)$.
Let $g\in\mathcal J$. Then $g=\sum_{i=1}^{m}A_iQ_i$ for suitable $A_i\in\Lx$. Being the right quotient of $H_i$ by the monic polynomial $P_i$, $Q_i$ lies in $\mathbb E_i[x_i;\sigma_i]$. Since $\sigma_i$ fixes $\mathbb E_k$ and $\sigma_k$ fixes $\mathbb E_i$, $Q_i$ commutes with $P_k$ whenever $k\neq i$, and we have
\[
Q_iF=P_1\cdots P_{i-1}(Q_iP_i)P_{i+1}\cdots P_m=
P_1\cdots P_{i-1}H_iP_{i+1}\cdots P_m\in H_\Lambda.
\]
Consequently, $gF=\sum_{i=1}^{m}A_i(Q_iF)\in H_\Lambda$.

Thus, $g\in(H_\Lambda:F)$, and hence $\mathcal J\subseteq(H_\Lambda:F)$.
It follows that
\begin{equation}\label{diseg colon ideal}
\dim_{\mathbb L}
\frac{\Lx}{(H_\Lambda:F)}
\le
\dim_{\mathbb L}
\frac{\Lx}{\mathcal J}.
\end{equation}

Since each $Q_i$ is monic and depends only on $x_i$, subtracting
suitable left multiples of $Q_i$ reduces the degree in $x_i$
without increasing the degrees in the other variables.
Applying this procedure successively for $i=1,\ldots,m$, every
polynomial in $\Lx$ is congruent modulo $\mathcal J$ to an
$\mathbb L$-linear combination of the monomials
\[
x_1^{a_1}\cdots x_m^{a_m},
\qquad
0\le a_i<\deg Q_i
\quad\text{for every }i.
\]
Their residue classes therefore span $\Lx/\mathcal J$ over
$\mathbb L$, and hence
\[
\dim_{\mathbb L}\frac{\Lx}{\mathcal J}
\le
\prod_{i=1}^{m}\deg Q_i.
\]
By construction,
\[
\deg Q_i
=
\begin{cases}
\nu_i, & i<s,\\
\nu_s-\ell, & i=s,\\
1, & i>s.
\end{cases}
\]
Therefore,
\[
\begin{aligned}
\dim_{\mathbb L}
\frac{\Lx}{\mathcal J}
&\le
\left(
\prod_{i=1}^{s-1}\nu_i
\right)
(\nu_s-\ell)
\left(
\prod_{i=s+1}^{m}1
\right)\\
&=
(\nu_s-\ell)
\prod_{i=1}^{s-1}\nu_i=d(r,\bm\nu).
\end{aligned}
\]
Recalling Remark \ref{rk=dimrispcolonideal} and 
Corollary \ref{corollario moore dickson e srk}, and using \eqref{diseg colon ideal} we conclude
\[
\srk(F)=\dim_{\mathbb L}\frac{\Lx}{(H_\Lambda:F)}
\le
d(r,\bm\nu).
\]
On the other hand, since $F\ne0$ and $\deg F=r$, Theorem \ref{thm:sum-rank-alon-furedi} gives $\srk(F)\ge d(r,\bm\nu).$

In conclusion, the lower bound is attained by  $F$, i.e., $\srk(F)=d(r,\bm\nu)$.
\end{proof}
\begin{example}[Continuation of Example~\ref{ex:running-example}]
Consider the code $\SRRM_{\mathbb L/\mathbb K}(2,2)$ of order $r=2$. In this case, the integers $s$ and $\ell$ appearing in the minimum-distance formula are $s=2$ and $\ell=r=2$. Therefore,
\[
d(2,(4,4))=(\nu_2-\ell)\nu_1=(4-2)\cdot4=8.
\]
Following the construction in the proof of the previous theorem, set $P_1(x_1)=1$ and choose $P_2(x_2)=x_2^2-1$ as the monic right divisor of $H_2(x_2)$. Thus, $\deg P_2=2=\ell$ and
\[
H_2(x_2)=Q_2(x_2)P_2(x_2),
\qquad\text{where}\quad
Q_2(x_2)=x_2^2+2.
\]
Moreover,
\[
H_1(x_1)=Q_1(x_1)P_1(x_1),
\qquad\text{where}\quad
Q_1(x_1)=H_1(x_1).
\]
The polynomial constructed in the proof is therefore $F=P_1P_2=x_2^2-1$ and it is a nonzero, reduced polynomial of degree $2$. Hence, it represents a codeword of
$\SRRM_{\mathbb L/\mathbb K}(2,2)$.

Consider the left ideal
\[
\mathcal J=\mathbb L[x_1,x_2;\sigma_1,\sigma_2]Q_1+\mathbb L[x_1,x_2;\sigma_1,\sigma_2]Q_2=
\mathbb L[x_1,x_2;\sigma_1,\sigma_2]H_1+\mathbb L[x_1,x_2;\sigma_1,\sigma_2](x_2^2+2).
\]
As seen in the proof of Theorem~\ref{thm:minimum-distance}, $\mathcal J\subseteq(H_\Lambda:F)$, and $\mathbb L[x_1,x_2;\sigma_1,\sigma_2]/\mathcal J$ is spanned by the classes of $x_1^{a_1}x_2^{a_2}$ with $a_1<\deg Q_1=4$ and $a_2<\deg Q_2=2$. Hence
\[
\srk(F)=\dim_{\mathbb L}\frac{\mathbb L[x_1,x_2;\sigma_1,\sigma_2]}{(H_{\Lambda}:F)}\le\dim_{\mathbb L}\frac{\mathbb L[x_1,x_2;\sigma_1,\sigma_2]}{\mathcal J}\le4\cdot2=8,
\]
while Theorem~\ref{thm:sum-rank-alon-furedi} gives $\srk(F)\ge(\nu_2-\ell)\nu_1=8$. This can also be checked directly: on the block indexed by $\bm j$, $\bm\Phi(F)$ equals $(\alpha_{2,j_2}\sigma_2)^2-1=\lambda_{2,j_2}-1$, which is $0$ for $j_2=1$ and $-3$ for $j_2=2$. Thus, $\bm\Phi(F)$ vanishes on two blocks and is invertible on the other two, and $d_{\min}\bigl(\SRRM_{\mathbb L/\mathbb K}(2,2)\bigr)=8$.
\end{example}

\subsection{Lower-set sum-rank Reed--Muller-type codes}
\label{subsec:lower-set-codes}

The proofs of Theorems~\ref{thm:sum-rank-alon-furedi} and~\ref{thm:minimum-distance} only use that the support is closed under componentwise decrease. We record the resulting extension, which is used in Section~\ref{subsec:comparison-LRM}.
\begin{notation}
For $\bm u=(u_1,\ldots,u_m),\bm v=(v_1,\ldots,v_m)\in\mathbb{N}^m$, we write $\bm v\leq \bm u$ if $v_i\leq u_i$ for every $i\in\{1,\ldots,m\}$.
\end{notation}
\begin{definition}
\label{def:lower-set}
A subset
\[
D\subseteq\Delta(\bm\nu)=\Delta(\nu_1)\times\cdots\times\Delta(\nu_m)
\]
is called a \textbf{lower set} if
\[
\bm u\in D,\quad \bm v\in\Delta(\bm\nu),\quad \bm v\leq \bm u
\qquad\Longrightarrow\qquad\bm v\in D.
\]
For a nonempty lower set $D\subseteq\Delta(\bm\nu)$, we define the
\textbf{lower-set sum-rank Reed--Muller-type code associated with $D$} as
\[
\SRRM_{\LK}(D):=\Span_\mathbb L\left\{\bm x^{\bm u} \bmod (H_1,\dots, H_m)\ \middle| \,\bm u\in D\right\}.
\]
\end{definition}
Since all the elements of $D$ belong to $\Delta(\bm\nu)$, the corresponding monomials are reduced representatives in $\Lx/H_\Lambda$. Consequently, their residue classes are linearly independent over $\mathbb L$, and
hence
\[
\dim_\mathbb L\SRRM_{\LK}(D)=\#D.
\]
We now show that the minimum sum-rank distance can also be determined exactly.

\begin{theorem}
\label{thm:lower-set-distance}
Let $D\subseteq\Delta(\bm\nu)$ be a nonempty lower set. Then
\[
d_{\min}\left(\SRRM_{\LK}(D)\right)=\min_{\bm u\in D}\prod_{i=1}^m(\nu_i-u_i).
\]
\end{theorem}

\begin{proof}
Set
\[
d(D,\bm\nu):=\min_{\bm u\in D}\prod_{i=1}^m(\nu_i-u_i).
\]
We first prove that every nonzero codeword has sum-rank weight at least $d(D,\bm\nu)$. Let $0\neq F\in\SRRM_{\LK}(D)$. Fix a monomial order refining the total degree and write $\LM(F)=\bm x^{\bm u},$ with $\bm u=(u_1,\ldots,u_m)\in D.$
The injectivity argument used in the proof of Theorem \ref{thm:sum-rank-alon-furedi}, applied to
\[
T_{\bm u}:=\Span_\mathbb L\left\{\bm x^{\bm v}\mid 0\leq v_i<\nu_i-u_i,\;i=1,\dots,m\right\}
\]
gives
\[
\srk(F)\geq \dim_\mathbb L T_{\bm u}=\prod_{i=1}^m(\nu_i-u_i)\geq d(D,\bm\nu).
\]
Hence, 
\[
d_{\min}\bigl(\SRRM_{\LK}(D)\bigr)\geq d(D,\bm\nu).
\]
It remains to prove that this bound is attained. Choose $\bm u=(u_1,\ldots,u_m)\in D$ such that
\[
d(D,\bm\nu)=\prod_{i=1}^m(\nu_i-u_i).
\]
By Lemma \ref{lem:existence monic right divisor}, for every $i\in\{1,\ldots,m\}$ there exists a monic right divisor $P_i(x_i)\in \mathbb E_i[x_i;\sigma_i]$ of $H_i(x_i)$ satisfying $\deg P_i=u_i$. Write $H_i(x_i)=Q_i(x_i)P_i(x_i)$ and define
\[
F_{\bm u}:=\prod_{i=1}^m P_i(x_i).
\]
Every exponent vector $\bm v$ occurring in $F_{\bm u}$ satisfies
$\bm v\leq\bm u$. Since $D$ is a lower set, $\mathrm{Supp}(F_{\bm u})\subseteq D,$ and therefore $F_{\bm u}\in\SRRM_{\LK}(D)$.

Finally, following the same reasoning and the colon-ideal argument used in the proof of Theorem \ref{thm:minimum-distance} we obtain
\[
\srk(F_{\bm u})\leq\prod_{i=1}^m\deg Q_i=\prod_{i=1}^m(\nu_i-u_i)=d(D,\bm\nu).
\]
Together with the lower bound, this proves the claim.
\end{proof}

Theorem~\ref{thm:lower-set-distance} contains the minimum-distance formula of
Theorem \ref{thm:minimum-distance} as a special case. Indeed, for
\[
D_r:=\left\{\bm u\in\Delta(\bm\nu) \ \middle|\ |\bm u|\leq r\right\},
\]
we have
\[
\SRRM_{\LK}(D_r)=\SRRM_{\LK}(r,m),
\]
while
\[
\min_{\bm u\in D_r}\prod_{i=1}^m(\nu_i-u_i)=d(r,\bm\nu).
\]

A particularly natural class of lower sets is obtained by imposing a weighted degree condition.

\begin{definition}
\label{def:weighted-srrm}
Let $\omega=(\omega_1,\ldots,\omega_m)\in\mathbb{R}_{>0}^m$ and let $c\in\mathbb Z_{\geq 0}$. Define
\[
D_{\omega}(c):=\left\{\bm u\in\Delta(\bm\nu) \ \middle|\ \sum_{i=1}^m\omega_i u_i\leq c\right\}.
\]
The associated code
\[
\SRRM_{\LK}^{\omega}(c,m):=\SRRM_{\LK}\bigl(D_{\omega}(c)\bigr)
\]
will be called a \textbf{weighted sum-rank Reed--Muller-type code}.
\end{definition}

Since all the weights $\omega_i$ are positive, the set $D_\omega(c)$ is a lower set. Hence, Theorem~\ref{thm:lower-set-distance} immediately gives the following.

\begin{corollary}
\label{cor:weighted-distance}
With the notation of Definition~\ref{def:weighted-srrm}, we have
\[
\dim_\mathbb L\SRRM_{\LK}^{\omega}(c,m)=\#\left\{\bm u\in\Delta(\bm\nu)
\ \middle|\ \sum_{i=1}^m\omega_i u_i\leq c\right\}
\]
and
\[
d_{\min}\left(\SRRM_{\LK}^{\omega}(c,m)\right)=\min\left\{\prod_{i=1}^m(\nu_i-u_i) \ \middle|\ \bm u\in\Delta(\bm\nu),\ \sum_{i=1}^m\omega_i u_i\leq c\right\}.
\]
\end{corollary}

\begin{remark}
\label{rem:weighted-total-degree}
For $\omega=(1,\ldots,1)$ and every admissible order $r$, taking $c=r$ gives $D_{\omega}(r)=D_r$.
Consequently,
\[
\SRRM^{\omega}_{\LK}(r,m)=\SRRM_{\LK}(r,m).
\]
Thus, the weighted construction contains the codes of Definition~\ref{def: srk reed muller codes} as a special case.

In the finite-field setting considered in Section \ref{subsec:comparison-LRM}, weighted supports of
the form
\[
    \frac{u_1}{n_1}+\cdots+\frac{u_m}{n_m}\leq c
\]
arise naturally from lattice simplices. We will return to this specialization
when comparing our construction with the linearized Reed--Muller codes of
Berardini and Caruso.
\end{remark}

\section{Comparison with previous constructions}\label{sec: comparisons}

The aim of this section is to position the family introduced above
with respect to some classical constructions in the Hamming, rank,
and sum-rank metrics. 
\subsection{Recovering classical code families}
By specializing the
parameters $N,t$ and $m$, we recover several known families of codes.

For each $i\in\{1,\ldots,m\}$, let
\[
\Lambda_i:=\{\lambda_{i,1},\ldots,\lambda_{i,t_i}\}\subseteq\mathbb K^*.
\]
It will be convenient to compare an isometrically equivalent version of our code defined as
\[
\mathcal C_r:=\bm\Phi\left(\SRRM_{\mathbb L/\mathbb K}(r,m)\right)\subseteq(\mathbb L[G])^t
\]
directly with previously known evaluation codes.
Table~\ref{tab:specializations} summarizes the main specializations.
\begin{table}[H]
\centering
\begin{tabularx}{\textwidth}{@{}c c X@{}}
\toprule
\textbf{Parameter choice}
&
\textbf{Metric}
&
\textbf{Resulting code family}
\\
\midrule
$m=1$
&
Sum-rank
&
Linearized Reed--Solomon codes with $t$ blocks of length $N$ \cite{M-P2018}
\\[2mm]

$t=1$
&
Rank
&
 $\bm{\sigma}$-Reed--Muller codes  \cite{ACLN}
\\[2mm]

$m=t=1$
&
Rank
&
Gabidulin codes over cyclic Galois extensions \cite{G85}
\\[2mm]

$N=1$
&
Hamming
&
affine Cartesian codes on grids in $(\mathbb K^*)^m$ \cite{GT2013,LRMV2014}
\\[2mm]

$N=m=1$
&
Hamming
&
Reed--Solomon codes with nonzero evaluation points \cite{RS60}
\\
\bottomrule
\end{tabularx}
\caption{Relevant specializations of the proposed construction.}
\label{tab:specializations}
\end{table}
These identifications displayed in Table~\ref{tab:specializations} are made precise in the following proposition.
\begin{proposition}\label{codes specializations}
Let
\[
\mathcal C_r:=\bm\Phi\left(\SRRM_{\mathbb L/\mathbb K}(r,m)\right)\subseteq(\mathbb L[G])^t
\]
be the code isometrically equivalent to $\SRRM_{\mathbb L/\mathbb K}(r,m)$ through the representation isomorphism $\bm\Phi$. Then the following statements hold:
\begin{enumerate}[label=\textup{(\arabic*)}]
    \item if $m=1$, then $\mathcal C_r$ is a linearized Reed--Solomon code with $t$ blocks, each of length $N$ \cite{M-P2018};

    \item if $t=1$, then $\mathcal C_r$ coincides with the
    $\bm\sigma$-Reed--Muller code of order $r$ introduced  in \cite{ACLN};

    \item if $N=1$, then $\mathcal C_r$ is an affine Cartesian
    evaluation code, endowed with the Hamming metric \cite{GT2013,LRMV2014}.
\end{enumerate}
\end{proposition}
\begin{proof}
\begin{enumerate}[label=\textup{(\arabic*)}]
    \item Suppose first that $m=1$. In this case $\mathbb L/\mathbb K$ is cyclic, with $G=\langle\sigma\rangle$. Fix a $\mathbb K$-basis $\mathcal B_{\mathrm{ev}}= (\beta_1,\ldots,\beta_N)$ of $\mathbb L$. For every $h\in\{0,\ldots,r\}$, the row obtained from the monomial $x^h$ has entries
\[
\N^h(\alpha_j)\sigma^h(\beta_k),
\qquad
1\leq j\leq t,
\quad
1\leq k\leq N.
\]
These are the entries of the generator matrix of the linearized Reed--Solomon code of \cite{M-P2018} with the same $\alpha_j$ and $\beta_k$, so $\ev_{\mathcal B_{\mathrm{ev}}^{(t)}}(\mathcal C_r)$ is that code. Moreover, $\dim_{\mathbb L}\mathcal C_r=r+1$, and for $0\le r<Nt-1$ Theorem~\ref{thm:minimum-distance} (with $m=1$, $\nu_1=Nt$, $s=1$, $\ell=r$) gives
\[
d_{\min}(\mathcal C_r)=Nt-r=Nt-(r+1)+1;
\]
for $r=Nt-1$ the code is the whole space and $d_{\min}=1$. Hence, $\mathcal C_r$ attains the sum-rank Singleton bound.
\item Assume now that $t=1$. Then $\nu_i=n_i$ for every $i$, while the ambient space consists of a single copy of $\mathbb L[G]$. Writing $\bm\alpha=(\alpha_{1,1},\ldots,\alpha_{m,1}),$ for every $\bm u=(u_1,\ldots,u_m)\in\Delta(\bm n)$ we have $\bm\Phi\left(x_1^{u_1}\cdots x_m^{u_m}\right)=\N^{\bm u}(\bm\alpha)\bm\sigma^{\bm u}.$
Since $\N^{\bm u}(\bm\alpha)\neq0$, multiplying each generator
by this scalar does not change its $\mathbb L$-linear span. Therefore,
\[
\mathcal C_r=\Span_{\mathbb L}\left\{\bm\sigma^{\bm u}\;\middle|\;\bm u\in\Delta(\bm n),\ |\bm u|\leq r\right\},
\]
which is exactly the rank-metric $\bm\sigma$-Reed--Muller code of \cite{ACLN}.
\item  Finally, assume that $N=1$. Then $\mathbb L=\mathbb K$, the Galois group is trivial, i.e., $n_i=1$ and $\sigma_i=\id$ for every $i$. Consequently, the skew polynomial ring reduces to the ordinary commutative polynomial ring $\mathbb K[x_1,\ldots,x_m]$ and
\[
\nu_i=t_i,\qquad H_i(x_i)=\prod_{j_i=1}^{t_i}(x_i-\lambda_{i,j_i}).
\]
Moreover, $\mathbb L[G]\cong\mathbb K$. Since $\mathbb E_i=\mathbb K$, the norm map $\N_{\mathbb E_i/\mathbb K}$ is the identity, and therefore $\lambda_{i,j}=\alpha_{i,j}.$ 

For every $\bm j=(j_1,\ldots,j_m)\in J$, the corresponding
component of $\bm\Phi$ is
\[
\bm\Phi_{\bm j}(F)=
\psi\left(
F(\alpha_{1,j_1}x_1,\ldots,\alpha_{m,j_m}x_m)
\right)=F(\alpha_{1,j_1},\ldots,\alpha_{m,j_m})=
F(\lambda_{1,j_1},\ldots,\lambda_{m,j_m}),
\]
because $\psi(x_i)=\sigma_i=\id_{\mathbb K}$. Hence,
\[
\mathcal C_r=\left\{\left(F(\lambda_{1,j_1},\ldots,\lambda_{m,j_m})\right)_{\bm j\in J}\;\middle|\;F\in\mathbb K[x_1,\ldots,x_m],\ \deg F\le r\right\}.
\]
Thus, $\mathcal C_r$ is precisely the affine Cartesian evaluation
code associated with the grid $\Lambda_1\times\cdots\times\Lambda_m$. Since every block has size $1$, the sum-rank metric reduces to the Hamming metric.
\end{enumerate}
\end{proof}    
\begin{remark}
The intersections of the parameter regimes considered in
Proposition \ref{codes specializations} recover two other classical code families.

First, if $m=t=1$, then $G=\langle\sigma\rangle$ is cyclic of order $N$, and 
\[
\mathcal C_r=\Span_{\mathbb L}\left\{\id,\sigma,\sigma^2,\dots,\sigma^r\right\}
\]
is a Gabidulin code of dimension $r+1$ over the cyclic Galois
extension $\mathbb L/\mathbb K$ \cite{G85}.

Second, if $N=m=1$, then $\mathbb L=\mathbb K$, $G$ is trivial, and the construction reduces to
\[
\mathcal C_r=\left\{\left(F(\lambda_1),\ldots,F(\lambda_t)\right)\;\middle|\;F\in\mathbb K[x],\ \deg F\le r\right\}.
\]
Since the elements $\lambda_1,\dots,\lambda_t$ are pairwise distinct and $0\le r\le t-1$, this is a Reed--Solomon code with parameters $[t,r+1,t-r]_{\mathbb K}$ \cite{RS60}.
\end{remark}

\subsection{Comparison with linearized Reed--Muller codes}
\label{subsec:comparison-LRM}

Berardini and Caruso introduced in \cite{BC2025} linearized Reed--Muller codes in the sum-rank metric, obtained by evaluating multivariate Ore polynomials over finite fields. We show that, after a monomial change of variables, each of their codes is the image under $\bm\Phi$ of a monomial code in one of our quotient algebras (Proposition~\ref{prop:BC-identification}). This determines the dimension of all their codes and reduces the comparison of the two families to a comparison of monomial supports. Whenever the support becomes a lower set, Theorem~\ref{thm:lower-set-distance} gives the exact minimum distance.

\subsubsection{The codes of Berardini and Caruso}

Let $\mathbb K=\F_q$ and $\mathbb L=\F_{q^N}$, and let $\varphi\colon a\mapsto a^q$ be the Frobenius automorphism, so that $G=\langle\varphi\rangle$. Fix $\e=(e_1,\ldots,e_m)\in\mathbb Z^m$ with $\gcd(e_1,\ldots,e_m,N)=1$, and let
\[
\mathcal A_{\e}:=\mathbb L[X_1^{\pm1},\ldots,X_m^{\pm1};\varphi^{e_1},\ldots,\varphi^{e_m}]
\]
be the Laurent Ore polynomial ring in which the variables commute and $X_ia=\varphi^{e_i}(a)X_i$ for $a\in\mathbb L$. For $\u\in\mathbb Z^m$ we write $\bm X^{\u}:=X_1^{u_1}\cdots X_m^{u_m}$, so that $\bm X^{\u}\bm X^{\v}=\bm X^{\u+\v}$ and $\bm X^{\u}a=\varphi^{\e\cdot\u}(a)\bm X^{\u}$, where $\e\cdot\u:=\sum_ie_iu_i$. The \textbf{central lattice}
\[
L_{\e}:=\{\u\in\mathbb Z^m\mid \e\cdot\u\equiv0\pmod N\}
\]
is the kernel of the surjective homomorphism $\mathbb Z^m\to\mathbb Z/N\mathbb Z$, $\u\mapsto\e\cdot\u$. Hence, $[\mathbb Z^m:L_{\e}]=N$ and $\mathbb Z^m/L_{\e}$ is cyclic.

The blocks are indexed by the characters $\gamma\in\widehat{L}_{\e}:=\Hom(L_{\e},\F_q^*)$. Every $\gamma$ extends to a map $\widetilde\gamma\colon\mathbb Z^m\to\mathbb L^*$ such that
\begin{equation}\label{eq:BC-cocycle}
\widetilde\gamma(\u+\v)=\widetilde\gamma(\u)\,\varphi^{\e\cdot\u}\bigl(\widetilde\gamma(\v)\bigr)\qquad\text{for all }\u,\v\in\mathbb Z^m
\end{equation}
\cite[Lemma~2.1]{BC2025}, and $\varepsilon_{\widetilde\gamma}(\bm X^{\u}):=\widetilde\gamma(\u)\varphi^{\e\cdot\u}$ extends $\mathbb L$-linearly to a ring homomorphism $\varepsilon_{\widetilde\gamma}\colon\mathcal A_{\e}\to\mathbb L[G]$. Here, as in Section~\ref{sec:preliminaries}, $\mathbb L[G]$ is identified with $\End_{\F_q}(\mathbb L)$. Fixing one extension for every $\gamma$, Berardini and Caruso consider $\varepsilon:=(\varepsilon_{\widetilde\gamma})_{\gamma\in\widehat{L}_{\e}}$. For a finite set $S\subseteq\mathbb Z^m$ we write
\[
\LRM(\e;S):=\varepsilon\bigl(\Span_{\mathbb L}\{\bm X^{\u}\mid\u\in S\}\bigr)\subseteq(\mathbb L[G])^{\widehat{L}_{\e}}.
\]
These codes have sum-rank length $N(q-1)^m$. In \cite{BC2025}, $S=C\cap\mathbb Z^m$ for a compact convex set $C\subseteq\mathbb R^m$. The \textbf{standard} codes are $\LRM(\e;c):=\LRM(\e;S)$ with $S=\{\u\in\mathbb N^m\mid|\u|\le c\}$ \cite[Definition~3.1]{BC2025}. The \textbf{lattice-simplex} codes are $\LRM(\e;cS_w):=\LRM(\e;cS_w\cap\mathbb Z^m)$, where, for a basis $w=(w_1,\ldots,w_m)$ of $L_{\e}$,
\[
S_w:=\Bigl\{\textstyle\sum_{i=1}^m\mu_iw_i\Bigm|\mu_i\ge0,\ \sum_{i=1}^m\mu_i\le1\Bigr\}.
\]
For $1\le c\le q-2$ the code $\LRM(\e;c)$ has dimension $\binom{m+c}{m}$, and for $0\le c\le q-2$ the code $\LRM(\e;cS_w)$ has dimension $\#(cS_w\cap\mathbb Z^m)$; in both cases the minimum distance is at least $N(q-1)^{m-1}(q-1-c)$ \cite[Theorems~3.1 and~3.3]{BC2025}.

\subsubsection{Adapted bases and the identification}

\begin{lemma}\label{lem:BC-adapted}
Let $b=(b_1,\ldots,b_m)$ be a basis of $\mathbb Z^m$, and put $\sigma_i:=\varphi^{\e\cdot b_i}$ and $n_i:=\ord(\sigma_i)$. The following are equivalent:
\begin{enumerate}[label=\textup{(\alph*)}]
\item $(n_1b_1,\ldots,n_mb_m)$ is a basis of $L_{\e}$;
\item $n_1\cdots n_m=N$;
\item $G=\langle\sigma_1\rangle\times\cdots\times\langle\sigma_m\rangle$.
\end{enumerate}
If they hold, the integers $n_1,\ldots,n_m$ are pairwise coprime.
\end{lemma}

\begin{proof}
Since $\sigma_i^{n_i}=\id_{\mathbb L}$, we have $n_ib_i\in L_{\e}$, so $L_0:=\bigoplus_{i}\mathbb Zn_ib_i\subseteq L_{\e}$ and $[\mathbb Z^m:L_0]=n_1\cdots n_m$. As $[\mathbb Z^m:L_{\e}]=N$, (a) and (b) are equivalent. If (a) holds, then $\mathbb Z^m/L_{\e}\cong\prod_i\mathbb Z/n_i\mathbb Z$ is cyclic, so the $n_i$ are pairwise coprime. Then the subgroups $\langle\sigma_i\rangle$ of the cyclic group $G$ have pairwise coprime orders with product $N=\#G$, which gives (c). Finally, (c) implies (b).
\end{proof}

We call such a basis $b$ \textbf{adapted} to $\e$. Two cases will be used.
\begin{enumerate}[label=\textup{(\roman*)}]
\item The standard basis $(\bm\delta_1,\ldots,\bm\delta_m)$ of $\mathbb Z^m$ is adapted if and only if the orders $n_i=\ord(\varphi^{e_i})$ are pairwise coprime with product $N$. In this case $L_{\e}=n_1\mathbb Z\bm\delta_1\oplus\cdots\oplus n_m\mathbb Z\bm\delta_m$, and we call $L_{\e}$ \textbf{rectangular}.
\item There is always an adapted basis with $(n_1,\ldots,n_m)=(1,\ldots,1,N)$. Indeed, since $\mathbb Z^m/L_{\e}$ is cyclic of order $N$, there is a basis $(v_1,\ldots,v_m)$ of $\mathbb Z^m$ such that $(v_1,\ldots,v_{m-1},Nv_m)$ is a basis of $L_{\e}$ (see the proof of \cite[Lemma~2.1]{BC2025}). Then $\varphi^{\e\cdot v_i}=\id$ for $i<m$. Moreover, if $\varphi^{\e\cdot v_m}$ had order $n<N$, then $nv_m\in L_{\e}$, which is impossible because $(v_1,\ldots,v_{m-1},Nv_m)$ is a basis of $L_{\e}$. Hence, the orders are $(1,\ldots,1,N)$, and $(v_1,\ldots,v_m)$ is adapted. In these variables, the ring is almost commutative in the sense of \cite[Section~3.4.1]{BC2025}.
\end{enumerate}
If $N$ is a prime power, then by Lemma~\ref{lem:BC-adapted} every adapted basis has orders $(1,\ldots,1,N)$ up to the order of the variables.

Fix an adapted basis $b$ with orders $\bm n=(n_1,\ldots,n_m)$. We apply our construction to the decomposition $G=\langle\sigma_1\rangle\times\cdots\times\langle\sigma_m\rangle$ of Lemma~\ref{lem:BC-adapted}, with $\mathbb E_i=\mathbb L^{\langle\sigma_1,\ldots,\widehat{\sigma_i},\ldots,\sigma_m\rangle}=\F_{q^{n_i}}$ and with the full norm grids $\Lambda_i=\F_q^*$. Thus, $t_i=q-1$, $\nu_i=n_i(q-1)$, and the blocks are indexed by $\lambda=(\lambda_1,\ldots,\lambda_m)\in(\F_q^*)^m$. Since the norm $\N_{\mathbb E_i/\F_q}\colon\mathbb E_i^*\to\F_q^*$ is surjective, we may choose $\alpha_{i,\mu}\in\mathbb E_i^*$ with $\N_{\mathbb E_i/\F_q}(\alpha_{i,\mu})=\mu$ for every $\mu\in\F_q^*$. With these choices,
\[
H_i(x_i)=\prod_{\mu\in\F_q^*}(x_i^{n_i}-\mu)=x_i^{n_i(q-1)}-1=x_i^{\nu_i}-1.
\]
The identity $H_i(x_i)=x_i^{\nu_i}-1$ is the main additional feature provided by our quotient description: every $x_i$ is a unit of $\Lx/H_\Lambda$, so $\bm x^{\bm w}$ is defined for every $\bm w\in\mathbb Z^m$, and $\bm x^{\bm w}=\bm x^{\overline{\bm w}}$, where $\overline{\bm w}\in\Delta(\bm\nu)$ denotes the componentwise reduction of $\bm w$ modulo $\bm\nu$. For $\u\in\mathbb Z^m$ let $[\u]_b\in\mathbb Z^m$ be its coordinate vector with respect to $b$, and for $S\subseteq\mathbb Z^m$ put
\[
\overline S_b:=\bigl\{\overline{[\u]_b}\bigm|\u\in S\bigr\}\subseteq\Delta(\bm\nu).
\]

\begin{proposition}\label{prop:BC-identification}
Let $b$ be an adapted basis, with the notation above. For every finite set $S\subseteq\mathbb Z^m$, the code $\LRM(\e;S)$ is sum-rank equivalent to
\[
\bm\Phi\Bigl(\Span_{\mathbb L}\bigl\{\bm x^{\bm w}\bigm|\bm w\in\overline S_b\bigr\}\Bigr)\subseteq(\mathbb L[G])^{(q-1)^m}.
\]
The equivalence sends the block indexed by $\gamma\in\widehat{L}_{\e}$ to the block indexed by $\lambda=\bigl(\gamma(n_1b_1),\ldots,\gamma(n_mb_m)\bigr)$, and on each block it is the conjugation by an element of $\mathbb L^*$.
\end{proposition}

\begin{proof}
Since every $x_i$ is a unit of $\Lx/H_\Lambda$, the assignment $a\bm X^{\u}\mapsto a\,\bm x^{[\u]_b}$ defines an $\mathbb L$-linear map $\vartheta_b\colon\mathcal A_{\e}\to\Lx/H_\Lambda$. It is a ring homomorphism: $[\u+\v]_b=[\u]_b+[\v]_b$ and $\bm\sigma^{[\u]_b}=\prod_i\varphi^{(\e\cdot b_i)[\u]_{b,i}}=\varphi^{\e\cdot\u}$, so that
\[
\vartheta_b\bigl((a\bm X^{\u})(c\bm X^{\v})\bigr)=a\varphi^{\e\cdot\u}(c)\,\bm x^{[\u]_b+[\v]_b}=\bigl(a\bm x^{[\u]_b}\bigr)\bigl(c\bm x^{[\v]_b}\bigr)
\]
for $a,c\in\mathbb L$.

Fix $\lambda\in(\F_q^*)^m$, put $\bm\alpha^{(\lambda)}:=(\alpha_{1,\lambda_1},\ldots,\alpha_{m,\lambda_m})$, and let $\bm\Phi_\lambda$ be the component of $\bm\Phi$ indexed by $\lambda$. For $\bm w\in\Delta(\bm\nu)$ we have $\bm\Phi_\lambda(\bm x^{\bm w})=\N^{\bm w}(\bm\alpha^{(\lambda)})\,\bm\sigma^{\bm w}$ with $\N^{\bm w}(\bm\alpha^{(\lambda)})\in\mathbb L^*$. Since $\bm x^{\bm w}=\bm x^{\overline{\bm w}}$ and $\bm\sigma^{\overline{\bm w}}=\bm\sigma^{\bm w}$, the image of $\bm x^{\bm w}$ is a nonzero scalar multiple of $\bm\sigma^{\bm w}$ for every $\bm w\in\mathbb Z^m$. Hence
\[
\bm\Phi_\lambda\bigl(\vartheta_b(\bm X^{\u})\bigr)=\widetilde\gamma_\lambda(\u)\,\varphi^{\e\cdot\u}
\]
for a function $\widetilde\gamma_\lambda\colon\mathbb Z^m\to\mathbb L^*$, and the multiplicativity of $\bm\Phi_\lambda\circ\vartheta_b$ is equivalent to \eqref{eq:BC-cocycle} for $\widetilde\gamma_\lambda$. If $\u=\sum_ik_in_ib_i\in L_{\e}$, then $[\u]_b=(k_1n_1,\ldots,k_mn_m)$, and, since $(\alpha_{i,\lambda_i}\sigma_i)^{n_i}=\N_{\mathbb E_i/\F_q}(\alpha_{i,\lambda_i})\id_{\mathbb L}=\lambda_i\id_{\mathbb L}$,
\[
\widetilde\gamma_\lambda(\u)\,\id_{\mathbb L}=\prod_{i=1}^m\bigl(\bm\Phi_\lambda(x_i)^{n_i}\bigr)^{k_i}=\Bigl(\prod_{i=1}^m\lambda_i^{k_i}\Bigr)\id_{\mathbb L}.
\]
Thus, $\widetilde\gamma_\lambda$ extends the character $\gamma_\lambda$ of $L_{\e}$ defined by $\gamma_\lambda(n_ib_i)=\lambda_i$. Since $(n_1b_1,\ldots,n_mb_m)$ is a basis of $L_{\e}$, the map $\lambda\mapsto\gamma_\lambda$ is a bijection $(\F_q^*)^m\to\widehat{L}_{\e}$. Therefore,  $\bm\Phi\circ\vartheta_b$ is the evaluation map of \cite{BC2025} for the extensions $\widetilde\gamma_\lambda$.

Let $\widetilde\gamma$ be any other extension of $\gamma_\lambda$ satisfying \eqref{eq:BC-cocycle}, and put $\mu:=\widetilde\gamma/\widetilde\gamma_\lambda$. Then $\mu(\u+\v)=\mu(\u)\varphi^{\e\cdot\u}(\mu(\v))$ and $\mu|_{L_{\e}}=1$; in particular $\mu(\u+\v)=\mu(\u)$ for $\v\in L_{\e}$. Hence, $\mu$ induces a $1$-cocycle $G\to\mathbb L^*$ via $\varphi^{\e\cdot\u}\mapsto\mu(\u)$. By Hilbert's Theorem~90 there is $\beta\in\mathbb L^*$ with $\mu(\u)=\varphi^{\e\cdot\u}(\beta)\beta^{-1}$ for all $\u$, and then $\varepsilon_{\widetilde\gamma}(\bm X^{\u})=\beta^{-1}\varepsilon_{\widetilde\gamma_\lambda}(\bm X^{\u})\,\beta$, where $\beta$ also denotes multiplication by $\beta$ (compare \cite[Remark~2.1]{BC2025}). Conjugation preserves the rank of each block, so $\LRM(\e;S)$ is sum-rank equivalent to $\bm\Phi\bigl(\vartheta_b(\Span_{\mathbb L}\{\bm X^{\u}\mid\u\in S\})\bigr)$, which is the code in the statement because $\vartheta_b(\bm X^{\u})=\bm x^{\overline{[\u]_b}}$.
\end{proof}

\begin{corollary}\label{cor:BC-dimension}
For every finite set $S\subseteq\mathbb Z^m$, $\dim_{\mathbb L}\LRM(\e;S)$ equals the number of classes of $S$ modulo the lattice $(q-1)L_{\e}$. In particular, $\dim_{\mathbb L}\LRM(\e;S)=\#S$ if and only if no two distinct elements of $S$ are congruent modulo $(q-1)L_{\e}$; this holds for $S=\{\u\in\mathbb N^m\mid|\u|\le c\}$ and for $S=cS_w\cap\mathbb Z^m$ whenever $c\le q-2$.
\end{corollary}

\begin{proof}
By Remark~\ref{rem:universal-grobner-basis} and Theorem~\ref{representation theorem}, distinct elements of $\Delta(\bm\nu)$ give linearly independent images under $\bm\Phi$, so the dimension is $\#\overline S_b$. Moreover, $\overline{[\u]_b}=\overline{[\u']_b}$ if and only if $[\u-\u']_b\in\bigoplus_i\nu_i\mathbb Z$, that is, $\u-\u'\in\bigoplus_i(q-1)n_i\mathbb Zb_i=(q-1)L_{\e}$. Since $(q-1)L_{\e}\subseteq(q-1)\mathbb Z^m$, two distinct elements of $\{\u\in\mathbb N^m\mid|\u|\le c\}$ are not congruent if $c\le q-2$. If $\u,\u'\in cS_w$, the coordinates of $\u-\u'$ with respect to $w$ lie in $[-c,c]$, whereas every nonzero element of $(q-1)L_{\e}$ has a coordinate of absolute value at least $q-1$ with respect to $w$.
\end{proof}

\subsubsection{Standard codes for rectangular central lattices}

Assume that $L_{\e}$ is rectangular, that is, the orders $n_i=\ord(\varphi^{e_i})$ are pairwise coprime with product $N$, and take $b$ to be the standard basis, so that $\sigma_i=\varphi^{e_i}$ and $\vartheta_b(\bm X^{\u})=\bm x^{\u}$.

\begin{proposition}\label{prop:BC-standard-all-degrees}
Assume that $L_{\e}$ is rectangular. For every integer $c\ge0$, the code $\LRM(\e;c)$ is sum-rank equivalent to $\bm\Phi\bigl(\SRRM_{\LK}(c,m)\bigr)$ with $\Lambda_i=\F_q^*$ for every $i$. More generally, let $\Lambda_1,\ldots,\Lambda_m\subseteq\F_q^*$ be nonempty, and let $\LRM_\Lambda(\e;c)$ be the restriction of $\LRM(\e;c)$ to the blocks indexed by the characters $\gamma$ with
\[
\bigl(\gamma(n_1\bm\delta_1),\ldots,\gamma(n_m\bm\delta_m)\bigr)\in\Lambda_1\times\cdots\times\Lambda_m.
\]
Then $\LRM_\Lambda(\e;c)$ is sum-rank equivalent to $\bm\Phi\bigl(\SRRM_{\LK}(c,m)\bigr)$.
\end{proposition}

\begin{proof}
For $\Lambda_i=\F_q^*$ apply Proposition~\ref{prop:BC-identification} with $S=\{\u\in\mathbb N^m\mid|\u|\le c\}$: reduction modulo $\bm\nu$ does not increase $|\u|$ and fixes $\Delta(\bm\nu)$, so $\overline S_b=\{\u\in\Delta(\bm\nu)\mid|\u|\le c\}$. For arbitrary $\Lambda_i$, the polynomial $H_i$ has nonzero constant term and lies in $\mathbb K[x_i^{n_i}]$, so $x_i$ is still a unit of $\Lx/H_\Lambda$. The proof of Proposition~\ref{prop:BC-identification}, restricted to the blocks $\lambda\in\Lambda_1\times\cdots\times\Lambda_m$, shows that the restricted code is sum-rank equivalent to the image under $\bm\Phi$ of the classes of the polynomials of degree at most $c$, which is $\bm\Phi(\SRRM_{\LK}(c,m))$ by Definition~\ref{def: srk reed muller codes}.
\end{proof}

\begin{corollary}\label{cor:BC-standard-all-parameters}
Assume that $L_{\e}$ is rectangular, let $c\ge0$, and let $\nu_i=n_i(q-1)$ be ordered so that $\nu_1\ge\cdots\ge\nu_m$; put $R:=\sum_{i=1}^m(\nu_i-1)$. If $c<R$, then $\LRM(\e;c)$ has dimension $k(c,\bm\nu)$ and minimum distance
\[
d_{\min}\bigl(\LRM(\e;c)\bigr)=d(c,\bm\nu)=(\nu_s-\ell)\prod_{i=1}^{s-1}\nu_i,
\]
where $c=\sum_{i=s+1}^m(\nu_i-1)+\ell$ with $0\le\ell<\nu_s-1$. If $c\ge R$, the code is the whole ambient space. The same holds for $\LRM_\Lambda(\e;c)$, with $\nu_i=n_i\,\#\Lambda_i$.
\end{corollary}

\begin{proof}
This follows from Proposition~\ref{prop:BC-standard-all-degrees}, Proposition~\ref{prop: dim srk muller codes} and Theorem~\ref{thm:minimum-distance}.
\end{proof}

We compare this with \cite[Theorem~3.1]{BC2025}. Let $n_{\min}:=\min_in_i$. For $c<n_{\min}(q-1)$, and in particular for $c\le q-2$, no degree restriction in a single variable is active, so
\[
\dim_{\mathbb L}\LRM(\e;c)=\binom{c+m}{m},
\qquad
d_{\min}\bigl(\LRM(\e;c)\bigr)=\bigl(n_{\min}(q-1)-c\bigr)\frac{N}{n_{\min}}(q-1)^{m-1}.
\]
For $1\le c\le q-2$, the ratio between this value and the bound $N(q-1)^{m-1}(q-1-c)$ of \cite[Theorem~3.1]{BC2025} is
\[
\frac{n_{\min}(q-1)-c}{n_{\min}(q-1-c)}\ \ge\ 1,
\]
with equality if and only if $n_{\min}=1$; for $c=q-2$ it equals $(q-1)-(q-2)/n_{\min}$. Since the $n_i$ are pairwise coprime, $n_{\min}>1$ requires $N$ to have at least $m$ distinct prime divisors. In particular, if $N$ is a prime power, the bound of \cite[Theorem~3.1]{BC2025} is exact in the rectangular case. Table~\ref{tab:BC-rectangular} lists some values with $n_{\min}>1$. For $m=2$, $\e=(r_1,r_2)$ with $r_1<r_2$ coprime and $r_1r_2=N$, we have $(n_1,n_2)=(r_2,r_1)$ and $d_{\min}=r_2(q-1)\bigl(r_1(q-1)-c\bigr)$ for $c<r_1(q-1)$. For $c\le q-2$ this equals $N(q-1)^2-r_2c(q-1)$, the value given by the improved bound of \cite[Remark~3.2]{BC2025}, which is therefore sharp.

\begin{table}[H]
\centering
\begin{tabular}{@{}cccccc@{}}
\toprule
$q$ & $\bm n$ & length & $c=1$ & $c=\lfloor(q-2)/2\rfloor$ & $c=q-2$\\
\midrule
$7$ & $(3,2)$ & $216$ & $198\ /\ 180$ & $180\ /\ 144$ & $126\ /\ 36$\\
$11$ & $(5,2)$ & $1000$ & $950\ /\ 900$ & $800\ /\ 600$ & $550\ /\ 100$\\
$13$ & $(3,2)$ & $864$ & $828\ /\ 792$ & $684\ /\ 504$ & $468\ /\ 72$\\
\bottomrule
\end{tabular}
\caption{Minimum distance of $\LRM(\e;c)$ for rectangular $L_{\e}$, compared with the bound of \cite[Theorem~3.1]{BC2025} (exact value\ /\ bound).}
\label{tab:BC-rectangular}
\end{table}

\subsubsection{Lattice-simplex codes}

We now turn to the more flexible construction introduced in \cite[Section~3.3]{BC2025}. For an adapted basis $b$ with orders $\bm n$, the vectors $w_i:=n_ib_i$ form a basis $w$ of $L_{\e}$ by Lemma~\ref{lem:BC-adapted}. For $c\ge0$ put
\[
D_c^{\bm n}:=\Bigl\{\u\in\Delta(\bm\nu)\Bigm|\frac{u_1}{n_1}+\cdots+\frac{u_m}{n_m}\le c\Bigr\}.
\]
Since the weights are positive, $D_c^{\bm n}$ is a lower set.

\begin{proposition}\label{prop:BC-simplex-all-degrees}
Let $b$ be an adapted basis with orders $\bm n$, and let $w=(n_1b_1,\ldots,n_mb_m)$. For every integer $c\ge0$, $\LRM(\e;cS_w)$ is sum-rank equivalent to $\bm\Phi\bigl(\SRRM_{\LK}(D_c^{\bm n})\bigr)$, with $\Lambda_i=\F_q^*$. Consequently,
\[
\dim_{\mathbb L}\LRM(\e;cS_w)=\#D_c^{\bm n},
\qquad
d_{\min}\bigl(\LRM(\e;cS_w)\bigr)=\min_{\u\in D_c^{\bm n}}\prod_{i=1}^m\bigl(n_i(q-1)-u_i\bigr).
\]
\end{proposition}

\begin{proof}
A vector $\u=\sum_i\mu_in_ib_i$ lies in $cS_w\cap\mathbb Z^m$ if and only if $[\u]_b=(\mu_1n_1,\ldots,\mu_mn_m)\in\mathbb N^m$ and $\sum_i\mu_i\le c$. Hence, $\{[\u]_b\mid\u\in cS_w\cap\mathbb Z^m\}=\{\bm a\in\mathbb N^m\mid\sum_ia_i/n_i\le c\}$. If $a_i=\overline a_i+k_i\nu_i$ with $0\le\overline a_i<\nu_i$, then $\overline a_i/n_i=a_i/n_i-k_i(q-1)\le a_i/n_i$. Therefore,  reduction modulo $\bm\nu$ maps this set onto $D_c^{\bm n}$, which it contains, and $\overline{(cS_w\cap\mathbb Z^m)}_b=D_c^{\bm n}$. The claims follow from Proposition~\ref{prop:BC-identification} and Theorem~\ref{thm:lower-set-distance}.
\end{proof}

Proposition~\ref{prop:BC-simplex-all-degrees} applies in the two cases above: for rectangular $L_{\e}$ with the canonical basis $w=(n_1\bm\delta_1,\ldots,n_m\bm\delta_m)$, and, for every $\e$, with $w=(v_1,\ldots,v_{m-1},Nv_m)$ and $\bm n=(1,\ldots,1,N)$. The latter case contains the almost-commutative family $\e=(0,\ldots,0,1)$, $w=(\bm\delta_1,\ldots,\bm\delta_{m-1},N\bm\delta_m)$ of \cite[Section~3.4.1]{BC2025}, as well as the codes that become almost commutative after a monomial change of variables, such as the one in \cite[Remark~3.3]{BC2025}. For small degrees, Proposition~\ref{prop:BC-simplex-all-degrees} takes a particularly simple form.

\begin{corollary}\label{cor:BC-simplex-small-degree}
In the situation of Proposition~\ref{prop:BC-simplex-all-degrees}, let $0\le c\le q-2$. Then
\[
\dim_{\mathbb L}\LRM(\e;cS_w)=\#(cS_w\cap\mathbb Z^m),
\qquad
d_{\min}\bigl(\LRM(\e;cS_w)\bigr)=N(q-1)^{m-1}(q-1-c).
\]
If $\bm n=(1,\ldots,1,N)$, the dimension equals $\binom{m+c}{m}+(N-1)\binom{m+c-1}{m}$.
\end{corollary}

\begin{proof}
The dimension is given by Corollary~\ref{cor:BC-dimension}. Put $Q:=q-1$ and $z_i:=u_i/n_i$ for $\u\in D_c^{\bm n}$. Then $z_i\ge0$, $\sum_iz_i\le c$, and $\prod_i(n_iQ-u_i)=N\prod_i(Q-z_i)$. For $a,b\ge0$ we have $(Q-a)(Q-b)=Q(Q-a-b)+ab\ge Q(Q-a-b)$, and repeated application gives $\prod_i(Q-z_i)\ge Q^{m-1}(Q-\sum_iz_i)\ge Q^{m-1}(Q-c)$. Equality is attained at $\u=(cn_1,0,\ldots,0)\in D_c^{\bm n}$, so the distance formula follows from Proposition~\ref{prop:BC-simplex-all-degrees}. If $\bm n=(1,\ldots,1,N)$, the dimension is the number of $\u\in\mathbb N^m$ with $u_1+\cdots+u_{m-1}+u_m/N\le c$. Write $u_m=Ny+s$ with $0\le s<N$. For $s=0$ the condition is $u_1+\cdots+u_{m-1}+y\le c$, which has $\binom{m+c}{m}$ solutions; for each $s\in\{1,\ldots,N-1\}$ it is $u_1+\cdots+u_{m-1}+y\le c-1$, which has $\binom{m+c-1}{m}$ solutions.
\end{proof}

For $c\le q-2$, the sharpness of the bound of \cite[Theorem~3.3]{BC2025} does not require our methods, and it holds for every $\e$ and every basis $w$ of $L_{\e}$. Indeed, since $w_1\in L_{\e}$, we have $\varepsilon_{\widetilde\gamma}(\bm X^{w_1})=\gamma(w_1)\id_{\mathbb L}$. Hence, for distinct $\xi_1,\ldots,\xi_c\in\F_q^*$, the codeword $\varepsilon\bigl(\prod_{k=1}^c(\bm X^{w_1}-\xi_k)\bigr)\in\LRM(\e;cS_w)$ vanishes on the $c(q-1)^{m-1}$ blocks with $\gamma(w_1)\in\{\xi_1,\ldots,\xi_c\}$ and is invertible on the others, so its weight is $N(q-1)^{m-1}(q-1-c)$. For lattice simplices, the new information in Proposition~\ref{prop:BC-simplex-all-degrees} therefore concerns the dimension formula of Corollary~\ref{cor:BC-dimension} and the range $c\ge q-1$. There the bound $N(q-1)^{m-1}(q-1-c)$ is no longer positive, and lattice points of $cS_w$ that are congruent modulo $(q-1)L_{\e}$ have the same evaluation. Hence, the dimension is $\#D_c^{\bm n}$ rather than the Ehrhart count $\#(cS_w\cap\mathbb Z^m)$, and the minimum over $D_c^{\bm n}$ in Proposition~\ref{prop:BC-simplex-all-degrees} is the exact minimum distance for every $c\ge0$; it equals $1$ once $D_c^{\bm n}=\Delta(\bm\nu)$. For instance, for $q=3$, $N=2$, $m=2$, $\e=(1,1)$, $v_1=(1,-1)$, $v_2=(0,1)$ and $w=(v_1,2v_2)$, the set $2S_w$ contains $9$ lattice points, whereas $\LRM(\e;2S_w)$ has dimension $\#D_2^{(1,2)}=7$ and minimum distance $2$. In particular, over a fixed field $\F_q$ the parameters of these codes are determined for every $c$, whereas \cite[Theorem~3.3]{BC2025} requires $q\geq c+2$, a restriction already noted in \cite[Section~3.3]{BC2025}.

\subsubsection{Hyperbolic supports}

For $c\le q-2$, the dimension of the almost-commutative codes of Corollary~\ref{cor:BC-simplex-small-degree} meets an upper bound for the number of lattice points of $cS_w$ \cite[Proposition~3.2 and the subsequent discussion]{BC2025}, so these codes have the largest dimension among the lattice-simplex codes of \cite{BC2025}. Theorem~\ref{thm:lower-set-distance} allows a comparison with all lower-set supports. For $\delta\ge1$ put
\[
\mathcal H_\delta:=\Bigl\{\u\in\Delta(\bm\nu)\Bigm|\prod_{i=1}^m(\nu_i-u_i)\ge\delta\Bigr\}.
\]
This is a lower set. By Theorem~\ref{thm:lower-set-distance}, $\SRRM_{\LK}(\mathcal H_\delta)$ has minimum distance at least $\delta$, and every lower set $D$ with $d_{\min}(\SRRM_{\LK}(D))\ge\delta$ is contained in $\mathcal H_\delta$, as for the classical hyperbolic codes \cite{GH2001}.

\begin{corollary}\label{cor:BC-hyperbolic}
Let $\bm n=(1,\ldots,1,N)$, $\Lambda_i=\F_q^*$, $0\le c\le q-2$ and $\delta:=N(q-1)^{m-1}(q-1-c)$. Then $D_c^{\bm n}\subseteq\mathcal H_\delta$, and $\bm\Phi(\SRRM_{\LK}(\mathcal H_\delta))$ has the same length and minimum distance as the almost-commutative code $\LRM(\e;cS_w)$ of Corollary~\ref{cor:BC-simplex-small-degree}, and dimension $\#\mathcal H_\delta\ge\#D_c^{\bm n}$.
\end{corollary}

\begin{proof}
The inclusion is the inequality in the proof of Corollary~\ref{cor:BC-simplex-small-degree}. Hence, $\min_{\u\in\mathcal H_\delta}\prod_i(\nu_i-u_i)=\delta$, since the minimum $\delta$ over $D_c^{\bm n}$ is attained in $\mathcal H_\delta$, and the claim follows from Theorem~\ref{thm:lower-set-distance} and Corollary~\ref{cor:BC-simplex-small-degree}.
\end{proof}

For small $c$ the two sets coincide; Table~\ref{tab:BC-hyperbolic} lists cases where the inclusion is strict. In the coordinates $y_i=\nu_i-u_i$, the set $\mathcal H_\delta$ consists of the integer points of the convex set $\{\bm y\in[1,\nu_1]\times\cdots\times[1,\nu_m]\mid y_1\cdots y_m\ge\delta\}$. Hence, by Proposition~\ref{prop:BC-identification}, the hyperbolic codes are themselves codes $\LRM(\e;C)$ for compact convex sets $C$, which are in general not simplices; their parameters are not covered by \cite{BC2025}.

\begin{table}[H]
\centering
\begin{tabular}{@{}cccccc@{}}
\toprule
$q$ & $N$ & $m$ & length & $(c,\ d_{\min})$ & $\#D_c^{\bm n}\to\#\mathcal H_\delta$\\
\midrule
$7$ & $2$ & $2$ & $72$ & $(4,24)$,\ $(5,12)$ & $25\to30$,\ $36\to48$\\
$7$ & $2$ & $3$ & $432$ & $(3,216)$,\ $(4,144)$,\ $(5,72)$ & $30\to35$,\ $55\to80$,\ $91\to176$\\
$8$ & $3$ & $2$ & $147$ & $(5,42)$,\ $(6,21)$ & $51\to64$,\ $70\to97$\\
$9$ & $2$ & $3$ & $1024$ & $(5,384)$,\ $(6,256)$,\ $(7,128)$ & $91\to133$,\ $140\to238$,\ $204\to458$\\
$11$ & $2$ & $2$ & $200$ & $(7,60)$,\ $(8,40)$,\ $(9,20)$ & $64\to81$,\ $81\to110$,\ $100\to149$\\
\bottomrule
\end{tabular}
\caption{Dimension of the almost-commutative lattice-simplex codes $\LRM(\e;cS_w)$ and of the hyperbolic codes $\SRRM_{\LK}(\mathcal H_\delta)$ with the same length and minimum distance $\delta$.}
\label{tab:BC-hyperbolic}
\end{table}

\subsubsection{Differences between the two families}

\begin{remark}\label{rem:BC-boundary}
By Proposition~\ref{prop:BC-identification}, every code $\LRM(\e;S)$ is equivalent to a monomial code in one of our quotient algebras. The two families differ in the following respects.
\begin{enumerate}[label=\textup{(\arabic*)}]
\item Base field. Our construction applies to every finite Abelian Galois extension, also over infinite fields and for non-cyclic Galois groups, such as $\mathbb Q(\sqrt2,\sqrt3)/\mathbb Q$.
\item Blocks. The codes of \cite{BC2025} use all characters of $L_{\e}$, whereas our norm grids $\Lambda_i\subseteq\mathbb K^*$ also give restrictions to products of subsets of blocks (Proposition~\ref{prop:BC-standard-all-degrees}). Since $\Lambda_i\subseteq\mathbb K^*$, the extended almost-commutative codes of \cite[Section~3.4.1]{BC2025}, which evaluate the commutative variables also at $0$, are not included in our framework.
\item Supports. Theorem~\ref{thm:lower-set-distance} gives the exact minimum distance of $\LRM(\e;S)$ when $\overline S_b$ is a lower set for some adapted basis $b$. This is not always the case. For $N=2$, $\e=(1,1)$ and the adapted basis $b=((1,-1),(0,1))$, we have $[\u]_b=(u_1,u_1+u_2)$, so for the standard code $\LRM(\e;c)$ with $1\le c\le q-2$ the set $\overline S_b=\{(a,s)\mid0\le a\le s\le c\}$ is not a lower set. Determining the minimum distance of codes whose reduced supports are not lower sets remains open.
\end{enumerate}
\end{remark}

In summary, over finite fields every linearized Reed--Muller code of \cite{BC2025} is a monomial code in our framework, and its dimension is given by Corollary~\ref{cor:BC-dimension} in every degree. When the support is a lower set in adapted coordinates, our results give the exact minimum sum-rank distance. This applies to the standard codes for rectangular central lattices and to the lattice-simplex codes associated with adapted bases, including the almost-commutative family, and extends \cite[Theorems~3.1 and~3.3]{BC2025} beyond the range $c\le q-2$. Finally, lower-set codes with hyperbolic supports have larger dimension than the almost-commutative lattice-simplex codes with the same length and minimum distance.

\section*{Acknowledgments}
This research has been partially supported by the Italian National Group for Algebraic and Geometric Structures and their Applications (GNSAGA -- INdAM). A. Neri is supported by the University of Naples Federico II through the FRA2024 Project CUP E63C26001530001  ``NHACT -- New Horizons in Algebraic Coding Theory''.

\bibliographystyle{abbrv}
\bibliography{bib}
\end{document}